\documentclass{lmcs}
\pdfoutput=1
\usepackage[utf8]{inputenc}

\usepackage{lastpage}
\lmcsdoi{17}{1}{4}
\lmcsheading{}{\pageref{LastPage}}{}{}%
{Sep.~10,~2019}{Jan.~22,~2021}{}

\usepackage{graphicx} 
\usepackage[linesnumbered,ruled,vlined]{algorithm2e}
\usepackage{amssymb}
\usepackage{amsmath}
\usepackage{threeparttable}

\newcommand{\lpmln}{LP\textsuperscript{MLN}{ }}
\newcommand{\lpmlnend}{LP\textsuperscript{MLN}}
\newcommand{\aspwc}{ASP\textsuperscript{wc}{ }}
\newcommand{\aspwcend}{ASP\textsuperscript{wc}}

\newcommand{\lglred}[2]{\left(\overline{{#1}_{#2}}\right)^{#2}}

\newcommand{\ase}[1]{SE^A(#1)}
\newcommand{\lse}[1]{SE^L(#1)}
\newcommand{\sse}[1]{SE^S(#1)}
\newcommand{\lue}[1]{UE^L(#1)}
\newcommand{\cse}[1]{SE^\times(#1)}

\newcommand{\assm}[1]{SM^A(#1)}
\newcommand{\lsm}[1]{SM^L(#1)}
\newcommand{\ssm}[1]{SM^S(#1)}
\newcommand{\psm}[1]{SM^P(#1)}
\newcommand{\csm}[1]{SM^\times(#1)}

\newcommand{\lpodtrans}[1]{\tau^\times(#1)}
\newcommand{\lpodtranssoft}[1]{\tau^\times_2(#1)}
\newcommand{\aspwctrans}[1]{\tau^c(#1)}
\newcommand{\aspwctranshard}[1]{\tau^c_h(#1)}
\newcommand{\aspwctranssoft}[1]{\tau^c_s(#1)}

\newcommand{\olit}[1]{lit\left(\overline{#1}\right)}
\newcommand{\oat}[1]{at\left(\overline{#1}\right)}

\newcommand{\ptrans}[2]{\Gamma(#1) \wedge \Delta(#2)}
\newcommand{\deltar}[1]{\delta_2(#1) \rightarrow \delta_1(#1)}

\begin{document}


\title[On the Strong Equivalences for \texorpdfstring{\lpmln}{LPMLN }Programs]{On the Strong Equivalences for \texorpdfstring{\lpmln}{LPMLN }Programs}
\titlecomment{
{\lsuper*}This paper is a thoughtful extension of \cite{Wang2019LPMLNSE}. 
Besides full proofs, this paper also adds three new parts: 
(1) investigating two different notions of strong equivalences for \lpmlnend; (2) analyzing the computational complexities of deciding strong equivalences; 
and (3) studying the relationships among the notions of strong equivalences for \lpmln and two logic formalisms: ASP with weak constraints and ordered disjunctions. 
}

\author[B. Wang]{Bin Wang}

\author[J. Shen]{Jun Shen\texorpdfstring{\textsuperscript{\textdagger}}{}}
\address{School of Computer Science and Engineering, Southeast University, Nanjing 211189, China}
\email{\{kse.wang,junshen,shutao\_zhang,seu\_zzz\}@seu.edu.cn}

\author[S. Zhang]{Shutao Zhang}

\author[Z. Zhang]{Zhizheng Zhang\texorpdfstring{\textsuperscript{\textdagger}}{}}
\thanks{\textsuperscript{\textdagger}{} Corresponding Author}

\begin{abstract}
\lpmln is a powerful knowledge representation and reasoning tool that combines the non-monotonic reasoning ability of Answer Set Programming (ASP) and the probabilistic reasoning ability of Markov Logic Networks (MLN). 
In this paper, we study the strong equivalence for \lpmln programs, 
which is an important tool for program rewriting and theoretical investigations in the field of logic programming.
First of all, we present the notion of p-strong equivalence for \lpmln and present a model-theoretical characterization for the notion. 
Then, we investigate several properties of the p-strong equivalence from the following four aspects. 
Firstly, we investigate two relaxed notions of the p-strong equivalence according to practical scenarios of program rewriting, 
and present corresponding characterizations for the notions. 
Secondly, we analyze the computational complexities of deciding strong equivalences for \lpmln programs. 
Thirdly, we investigate the relationships among the strong equivalences of \lpmln and two extensions of ASP: ASP with weak constraints and ordered disjunctions. 
Finally, we investigate \lpmln program simplification via the p-strong equivalence and present some syntactic conditions that decide the p-strong equivalence between a single \lpmln rule and the empty program. 
The contributions of the paper are as follows. 
Firstly, all the results presented in this paper provide a better understanding of \lpmln programming, which helps us further explore the properties of \lpmlnend. 
Secondly, the relationships among the strong equivalences open a way to study the strong equivalences for some logic formalisms by translating into \lpmlnend. 
Thirdly, the program simplification can be used to enhance the implementations of the \lpmln solvers, which is expected to facilitate the applications of \lpmlnend. 
\end{abstract}
\maketitle

\section{Introduction}
\lpmln \cite{Lee2016Weighted}, a new knowledge representation and reasoning language, 
is designed to handle non-monotonic, uncertain, and inconsistent knowledge by combining the logic programming methods of Answer Set Programming (ASP) \cite{Gelfond1988theSM,Brewka2011ASP} and Markov Logic Networks (MLN) \cite{Richardson2006mln}. 
Specifically, an \lpmln program can be viewed as a weighted ASP program, 
where each ASP rule is assigned a weight denoting its certainty degree,  
and each weighted rule is allowed to be violated by a set of beliefs associated with the program. 
For example, ASP rule ``$\leftarrow a, b. $'' is a constraint denoting the facts $a$ and $b$ are contrary, 
therefore, a belief set $X=\{a, b\}$ is invalid in the context of ASP programs containing the constraint. 
By contrast, in the context of \lpmlnend, above constraint becomes a weighted rule ``$w ~:~ \leftarrow a, b. $'', where $w$ denotes the certainty degree of the constraint. 
And the set $X = \{a, b\}$ becomes a valid belief set with loss of the certainty degree $w$. 
Usually, two kinds of certainty degrees are introduced to evaluate a belief set w.r.t. an \lpmln program. 
The weight degree is evaluated by the weights of rules satisfied by a belief set, and the probability degree can be viewed as a normalized weight degree. 
Based on the probabilistic programming of \lpmlnend, several inference tasks are introduced such as computing the marginal probability distribution of beliefs and computing the most probable belief sets. 

Due to its powerful expressivity, \lpmln has been broadly studied. 
On the practical side, 
\lpmln is suitable for applications that contain uncertain and inconsistent data. 
For example, in the tasks of classifying visual objects, \lpmln rules can be used to encode soft constraints among unlabeled objects such as ``objects equipped with wheels are usually cars'' \cite{Eiter2016Exploiting}. 
On the theoretical side, recent work on \lpmln  aims at establishing the relationships among \lpmln and other logic formalisms \cite{Lee2016Weighted,Balai2016realtionship,Lee2017lpmln}, 
developing \lpmln solvers \cite{Lee2017ComputingLpmln,Wang2017ParallelLpmln,Wu2018LPMLNModels}, 
acquiring the weights of rules automatically \cite{Lee2018WeightLearning}, 
exploring the properties of \lpmln \cite{Wang2018Splitting} etc.
Although all of the above results lay the foundation for knowledge representation and reasoning via \lpmlnend, 
there are many important theoretical problems have not been investigated for \lpmlnend. 
The investigation of strong equivalences for \lpmln is one of such kinds of problems. 

The notion of strong equivalence for ASP and its several extensions have been studied extensively, 
due to the fact that it is important for program rewriting in the field of logic programming \cite{Eiter2004SimplifyingLP,Puhrer2009casting,Woltran2011Equivalence}. 
Generally speaking, two ASP programs $\Pi_1$ and $\Pi_2$ are strongly equivalent, 
iff for any ASP program $\Pi_3$, the extended programs $\Pi_1 \cup \Pi_3$ and $\Pi_2 \cup \Pi_3$ have the same stable models \cite{Lifschitz2001Strongly,Turner2001SE}. 
Therefore, a logic program $\Pi_1$ can be rewritten as one of its strong equivalents $\Pi_2$ without considering its context. 
If the program $\Pi_2$ is easier to solve, it can be used to simplify the program $\Pi_1$, 
which is useful in implementations of solvers. 
For example, an ASP rule is called \textit{redundant} if it is strongly equivalent to the empty program. 
And if the positive and negative bodies of an ASP rule have common atoms, it is a redundant rule  \cite{Osorio2001Equivalence,Inoue2004EqUpdate,Lin2005Discover}. 
Obviously, eliminating redundant rules is an effective approach to enhancing ASP solvers. 
Inspired by the works of ASP, we believe studying strong equivalences for \lpmln programs will provide a theoretical tool for further investigations of \lpmlnend. 

In this paper, we study the strong equivalences for \lpmln programs. 
Usually, the notion of strong equivalence is based on the identity between inference results of logic programs, 
therefore, we firstly define the notion of strong equivalence on stable models and their probability distributions, called p-strong equivalence. 
Specifically, two \lpmln programs $P$ and $Q$ are p-strongly equvialent, 
if for any \lpmln program $R$, the extended program $P \cup R$ and $Q \cup R$ have the same stable models, 
and for any stable models $X$ of the extended programs, 
the probability degrees $Pr(P \cup R, X)$ and $Pr(Q \cup R, X)$ of $X$ w.r.t. the extended programs are the same, i.e. $Pr(P \cup R, X) = Pr(Q \cup R, X)$. 
Moreover, we present a relaxed notion of the p-strong equivalence, called the semi-strong equivalence, 
i.e. \lpmln programs $P$ and $Q$ are semi-strongly equivalent, if for any \lpmln program $R$, the extended programs $P \cup R$ and $Q \cup R$ have the same stable models. 
Then, we present a characterization for the semi-strong and p-strong equivalences by generalizing the strong equivalence models (SE-models) of ASP, 
which serves as a basic framework for the further investigations of the strong equivalence for \lpmlnend. 
Thirdly, we attempt to investigate a kind of intermediate notion of the p-strong and semi-strong equivalence, since both of the p-strong and semi-strong equivalence are not ideal for applications. 
As we know, the p-strong equivalence requires that stable models and their probability distributions of \lpmln programs are the same, while the semi-strong equivalence only requires the same stable models. 
By contrast, the conditions of p-strong equivalence are somewhat strict, 
but the conditions of semi-strong equivalence are too simple. 
Therefore, we consider a slightly relaxed notion of the p-strong equivalence, i.e. the qualitatively strong equivalence (q-strong equivalence). 
For the q-strong equivalence, the probability distributions of stable models of the extended programs are not required to be the same, 
but they are required to be order-preserving. 
Formally speaking, two \lpmln programs $P$ and $Q$ are q-strongly equivalent, if for any \lpmln program $R$, the extended program $P \cup R$ and $Q \cup R$ have the same stable models,  and for any stable models $X$ and $Y$ of the extended programs, $Pr(P \cup R, X) \leq Pr(P \cup R, Y)$ iff $Pr(Q \cup R, X) \leq Pr(Q \cup R, Y)$. 
Unfortunately, our results show that the q-strong equivalence is equivalent to the p-strong equivalence, i.e. two \lpmln programs are q-strongly equivalent iff they are q-strongly equivalent. 
Therefore, how to relax the strict condition of p-strong equivalence is still an open problem. 

Besides the strong equivalences notions presented in this paper, 
it is worth noting that Lee and Luo also investigated the strong equivalence for \lpmln programs, 
and presented an ASP-based implementation to check the strong equivalence  \cite{Lee2019LPMLNSE}.  
Both Lee and Luo's work and the early work of this paper \cite{Wang2019LPMLNSE} are presented at the  ICLP 2019. 
Although the definitions and characterizations of strong equivalences in these two works are similar, 
Lee and Luo also present four different kinds of characterizations for the semi-strong equivalence (i.e. the structural equivalence in \cite{Lee2019LPMLNSE}), 
which provides some different understandings of the strong equivalence for \lpmln programming. 
In this paper, we present a  formal comparison between the notions and characterizations of these different strong equivalences.

Next, we investigate several properties of the p-strong equivalence from four aspects. 
Firstly, we consider two special scenarios of problems solving. 
Under the scenarios, the p-strong equivalence is not suitable, 
therefore, we introduce two relaxed notions of the p-strong equivalence, i.e. the p-strong equivalence under the soft stable model semantics of \lpmln (sp-strong equivalence) and the p-strong equivalence under factual extensions (p-uniform equivalence). 
For the newly introduced notions, we investigate their characterizations by presenting corresponding SE-models. 
Secondly, we analyze the computational complexities of deciding strong equivalences for \lpmln programs. 
It shows that deciding all of the semi-strong, p-strong, and sp-strong equivalences are co-NP-complete, 
and deciding the uniform equivalence on stable models (semi-uniform equivalence) is in $\Pi_2^p$. 
To prove the co-NP-membership of deciding the p-strong equivalence, we present a translation from \lpmln programs to propositional formulas and show that two \lpmln programs are semi-strongly equivalent iff a related propositional formula is a tautology. 
According to Lin and Chen's work on the strong equivalence of ASP \cite{Lin2005Discover}, the translation can be used to discover syntactic conditions that decide the semi-strong equivalence of some classes of \lpmln programs. 
Thirdly, we investigate the relationships among the strong equivalences of \lpmln and two important extensions of ASP, 
i.e. ASP with weak constraints \cite{Calimeri2012ASPcore} and ASP with ordered disjunction \cite{Brewka2002LPOD}. 
The relationships show that the strong equivalence of some logic formalisms can be studied by translating them into \lpmln programs and using the results of strong equivalences for \lpmlnend. 
Finally, we use the notion of p-strong equivalence to simplify \lpmln programs and enhance \lpmln solvers. 
To decide the p-strong equivalence efficiently, we present a sufficient and necessary condition to characterize the strong equivalence between a single \lpmln rule and the empty program, i.e. the redundant \lpmln rules.

\section{Preliminaries}
\label{sec:preliminaries}
In this section, we firstly review the syntax and semantics of ASP and \lpmlnend. 
Then, we review the strong equivalence for ASP programs. 

\subsection{Syntax}
An ASP program is a finite set of rules of the form 
\begin{equation}
\label{eq:asp-rule-form}
l_1 ~ \vee ~ ... ~\vee ~ l_k ~\leftarrow~ l_{k+1}, ..., ~l_m, ~not~ l_{m+1}, ...,~not ~ l_n.
\end{equation}
where $l_i$s $(1 \leq i \leq n)$ are literals, 
$\vee$ is epistemic disjunction, 
and $not$ is default negation. 
A literal is either an atom $a$ or its negation $\neg a$, where $\neg$ is classical negation. 
For an ASP rule $r$ of the form \eqref{eq:asp-rule-form}, 
the sets of literals occurred in head, positive body, and negative body of $r$ are denoted by $h(r) = \{l_i ~|~ 1 \leq i \leq k\}$, 
$b^+(r)=\{l_i ~|~ k+1 \leq i \leq m\}$, 
and $b^-(r)=\{l_i ~|~ m+1 \leq i \leq n\}$, respectively. 
By $lit(r) = h(r) \cup b^+(r) \cup b^-(r)$, we denote the set of literals occurred in a rule $r$, 
by $at(r)$, we denote the set of atoms occurred in a rule $r$, i.e. $at(r) = \{a ~|~ a \in lit(r) \text{ or } \neg a \in lit(r)\}$, 
and by $lit(\Pi) = \bigcup_{r \in \Pi} lit(r)$ and $at(\Pi) = \bigcup_{r \in \Pi} at(r)$, we denote the sets of literals and atoms occurred in an ASP program $\Pi$ respectively. 
Therefore, an ASP rule $r$ of the form (\ref{eq:asp-rule-form}) can also be abbreviated as 
\begin{equation}
\label{eq:asp-rule-form-abbr}
h(r) \leftarrow b^+(r), ~not ~ b^-(r).
\end{equation}
An ASP rule is called a \textit{fact}, if both of its positive and negative bodies are empty, and it is called a \textit{constraint}, if its head is empty. 
An ASP program is called \textit{ground}, if it contains no variables. 

\begin{exa}
	\label{ex:asp-rule}
	Consider an ASP program $\Pi$
	\begin{eqnarray}
		& a \vee b. & \label{ex:asp-rule-fact}\\
		& \leftarrow a, b.  & \label{ex:asp-rule-constr} \\
		& a \leftarrow not ~ b. & \label{ex:asp-rule-normal} 
	\end{eqnarray}
	where rule \eqref{ex:asp-rule-fact} is a fact, rule \eqref{ex:asp-rule-constr} is a constraint, and rule \eqref{ex:asp-rule-normal} is a general ASP rule. 
	Since there is no variable in the program $\Pi$, it is a ground ASP program. 
\end{exa}

An \lpmln program is a finite set of weighted ASP rules $w:r$, 
where $w$ is the weight of rule $r$, and $r$ is an ASP rule of the form \eqref{eq:asp-rule-form}. 
The weight $w$ of an \lpmln rule is either a real number or a symbol ``$\alpha$'' denoting ``infinite weight'', 
and if $w$ is a real number, the rule is called \textit{soft}, otherwise, it is called \textit{hard}. 
Note that the weight of a soft rule could be any real number including positive number, negative number, and zero, 
which will be discussed in the semantics of \lpmlnend. 
An \lpmln rule $w:r$ is also called a \textit{weighted fact} or \textit{weighted constraint}, 
if $r$ is an ASP fact or constraint, respectively. 
For an \lpmln program $P$, 
we use $\overline{P}$ to denote the set of unweighted ASP counterpart of $P$, 
i.e. $\overline{P} = \{r ~|~ w:r \in P \}$. 
By $P^s$ and $P^h$, we denote the sets of all soft rules and hard rules in $P$, respectively. 
An \lpmln program $P$ is called ground, if its unweighted ASP counterpart $\overline{P}$ is ground. 
Usually, a non-ground logic program is considered as a shorthand for the corresponding ground program, 
therefore, we only consider ground logic programs in this paper. 

\begin{exa}
	\label{ex:lpmln-rule}
	Following \lpmln program $P$ is obtained from the program $\Pi$ in Example \ref{ex:asp-rule} by assigning weights. 
	\begin{eqnarray}
		\alpha &~:~& a \vee b. \label{ex:lpmln-rule-fact}\\
		2 &~:~& \leftarrow a, b. \label{ex:lpmln-rule-constr} \\
		-1 &~:~& a \leftarrow not ~ b. \label{ex:lpmln-rule-normal} 
	\end{eqnarray}
	Similarly, rule \eqref{ex:lpmln-rule-fact} is a weighted fact, 
	rule \eqref{ex:lpmln-rule-constr} is a weighted constraint, 
	rule \eqref{ex:lpmln-rule-normal} is a general \lpmln rule, 
	and the program $P$ is a ground \lpmln program. 
	In particular, rule \eqref{ex:lpmln-rule-fact} is a hard rule, 
	and other rules of $P$ are soft rules. 
	Moreover, it is obvious that the unweighted ASP counterpart of the program $P$ is the program $\Pi$ in Example \ref{ex:asp-rule}, i.e. $\overline{P} = \Pi$.
\end{exa}

\subsection{Semantics}
A ground set $I$ of literals is called \textit{consistent}, if there are no contrary literals occurred in $I$, i.e. for any literal $l$ of $I$, $\neg l$ does not occur in $I$. 
A consistent set $I$ of literals is usually called an \textit{interpretation}. 
For an ASP rule $r$ and an interpretation $I$, the satisfaction relation is defined as follows
\begin{itemize}
	\item $I$ satisfies the positive body of $r$, denoted by $I \models b^+(r)$, if $b^+(r) \subseteq I$;
	\item $I$ satisfies the negative body of $r$, denoted by $I \models b^-(r)$, if $b^-(r) \cap I = \emptyset$;
	\item $I$ satisfies the body of $r$, denoted by $I \models b(r)$, if $I \models b^+(r)$ and $I \models b^-(r)$;
	\item $I$ satisfies the head of $r$, denoted by $I \models h(r)$, if $h(r) \cap I \neq \emptyset$; and
	\item $I$ satisfies the rule $r$, denoted by $I \models r$, if $I \models b(r)$ implies $I \models h(r)$.
\end{itemize}
For an ASP program $\Pi$, an interpretation $I$ satisfies $\Pi$, denoted by $I \models \Pi$, if $I$ satisfies all rules of $\Pi$. 
For an ASP program $\Pi$ containing no default negations, an interpretation $I$ is a \textit{stable model} of $\Pi$, 
if $I \models \Pi$ and there does not exist a proper subset $I'$ of $I$ such that $I' \models \Pi$. 
For an arbitrary ASP program $\Pi$, the \textit{Gelfond-Lifschitz reduct} (GL-reduct) $\Pi^I$ of $\Pi$ w.r.t. an interpretation $I$ is defined as 
\begin{equation}
	\Pi^I = \{ h(r) \leftarrow b^+(r).  ~|~  r \in \Pi \text{ and } b^-(r) \cap I = \emptyset  \}
\end{equation}
An interpretation $I$ is a stable model of an arbitrary ASP program $\Pi$ if $I$ is a stable model of $\Pi^I$. 
By $\assm{\Pi}$, we denote the set of all stable models of an ASP program $\Pi$. 
\begin{exa}
Continue Example \ref{ex:asp-rule}, we consider four different interpretations, 
i.e. $I_1 = \emptyset$, $I_2 = \{a\}$, $I_3 = \{b\}$, and $I_4 = \{a, b \}$. 
For the interpretations $I_1$ and $I_4$, they do not satisfy all rules of the program $\Pi$, 
i.e. $I_4$ does not satisfy rule \eqref{ex:asp-rule-constr}, 
and $I_1$ does not satisfy the other two rules of the program. 
And for the interpretation $I_2$, we have $\Pi^{I_2} = \{ a \vee b. ~ \leftarrow a, b. ~ a.  \}$. 
It is easy to check that $I_2$ is a stable model of $\Pi$. 
Similarly, one can check that $I_3$ is also a stable model of the program $\Pi$, 
therefore, we have  $\assm{\Pi} = \{I_2, I_3 \}$. 
\end{exa}

An \lpmln  rule $w:r$ is satisfied by an interpretation $I$, 
denoted by $I \models w:r$, if $I\models r$. 
An \lpmln program $P$ is satisfied by $I$, denoted by $I \models P$, if $I$ satisfies all rules in $P$.
By $P_I$, we denote the set of rules of an \lpmln program $P$ that can be satisfied by an interpretation $I$, 
called the \textit{\lpmln reduct} of  $P$ w.r.t. $I$, i.e. $P_I=\{w:r \in P ~|~ I \models w:r\}$. 
An interpretation $I$ is a stable model of an \lpmln program $P$ if $I$ is a stable model of the ASP program $\overline{P_I}$. 
By $\lsm{P}$, we denote the set of all stable models of an \lpmln program $P$. 
For an \lpmln program $P$, the weight degree $W(P)$ of $P$ is defined as
\begin{equation}
W(P) = exp\left(\sum_{w:r \in P } w\right)
\end{equation}
Note that the weight $\alpha$ of a hard rule is a symbol, therefore, the weight degree $W(P)$ should be understood as a symbolic expression. 
For example, suppose $P$ is an \lpmln program such that $P = \{\alpha : r_1; ~ \alpha : r_2; ~ 3 : r_3\}$, 
the weight degree of $P$ is $W(P) = exp(2\alpha + 3)$.
For an \lpmln program $P$ and an interpretation $I$, 
by $h(P, I)$, we denote the number of hard rules of $P$ satisfied by $I$, 
i.e. $h(P, I) = |P^h_I|$, 
the \textit{weight degree} $W(P,I)$ of $I$ w.r.t. $P$ is defined as 
\begin{equation}
\label{eq:weight-sm}
W(P,I) = W(P_I) = exp\left(\sum_{w:r \in P_I } w\right)
\end{equation}
and the \textit{probability degree} $Pr(P,I)$ of $X$ w.r.t. $P$ is defined as
\begin{equation}
\label{eq:probability-sm}
Pr(P,I) = \begin{cases}
\lim\limits_{\alpha \rightarrow \infty} \frac{W(P,I)}{\Sigma_{I'\in \lsm{P}}W(P,I')} & \text{ if } I \in \lsm{P};\\
0 & \text{ otherwise. }
\end{cases}
\end{equation}
For a literal $l$, the \textit{probability degree} $Pr(P,l)$ of $l$ w.r.t. $P$ is defined as 
\begin{equation}
\label{eq:probability-lit}
Pr(P,l) = \sum_{l \in I \text{ and }  I \in \lsm{P}} Pr(P,I)
\end{equation}
A stable model $I$ of an \lpmln program $Q$ is called a \textit{probabilistic stable model} of $Q$ if $Pr(Q,I) \neq 0$.
By $\psm{Q}$, we denote the set of all probabilistic stable models of $Q$. 
It is easy to check that $I$ is a probabilistic stable model of $Q$, iff $I$ is a stable model of $Q$ that satisfies the most number of hard rules. 
Thus, for a stable model $I$ of $Q$, the probability degree of $I$ can be reformulated as follows
\begin{equation}
\label{eq:probability-psm}
Pr(Q,I) = \begin{cases}
\frac{W(Q^s,I)}{\Sigma_{I'\in \psm{Q}}W(Q^s,I')} & \text{ if } I \in \psm{Q};\\
0 & \text{ otherwise. }
\end{cases}
\end{equation}
Based on the above definitions, there are two kinds of main inference tasks for an \lpmln program $P$ \cite{Lee2017ComputingLpmln}:  
\begin{itemize}
	\item[-] Maximum A Posteriori (MAP) inference: compute the stable models with the highest weight or probability degree of the program $P$, i.e. the most probable stable models;
	\item[-] Marginal Probability Distribution (MPD) inference: compute the probability degrees of a set of literals w.r.t. the program $P$.
\end{itemize}

\begin{exa}
\label{ex:lpmln-inference}
Continue Example \ref{ex:lpmln-rule}, it is easy to check that the program $P$ has three stable models, i.e. $\emptyset$, $\{a\}$, and $\{b\}$, which is shown in Table \ref{tab:ex-lpmln-rule-sm}. 
For the interpretation $I = \{a, b\}$, we have $\overline{P_I} = \{a \vee b. ~ a \leftarrow not ~ b. \}$ and $\lglred{P}{I} = \{a \vee b. \}$. 
Since both of $\{a\}$ and $\{b\}$ can satisfy $\lglred{P}{I}$, 
$I$ is not a stable model of the \lpmln program $P$. 
Besides, the MAP and MPD inference results can be obtained by checking Table \ref{tab:ex-lpmln-rule-sm}, i.e. $\{a\}$ and $\{b\}$ are the most probable stable models of $P$, and we have $Pr(P, a) = Pr(P, b) = 0.5$.
	\begin{table}
		\caption{Stable Models of the Program $P$ in Example \ref{ex:lpmln-rule}}
		\label{tab:ex-lpmln-rule-sm}
    \def\arraystretch{1.2} 
	\begin{tabular}{c c c c c }
		\hline
		stable model $I$ & $P_I$ & $W(P,I)$ & $h(P, I)$ & $Pr(P,I)$ \\ 
		\hline 
		$\emptyset$ & $\{2 ~:~  \leftarrow a, b. \}$  & $e^2$ & $0$ & $0$  \\ 
		$\{a\}$ & $P$ & $e^{\alpha + 1}$ & $1$ &  $0.5$  \\ 
		$\{b\}$ & $P$ & $e^{\alpha + 1}$ & $1$  & $0.5$  \\ 
		\hline 
	\end{tabular} 
	\end{table}
\end{exa}

According to the syntax of \lpmlnend, the weight of an \lpmln rule could be a symbol $\alpha$, a positive number, a negative number, or zero. 
Here, we present a brief discussion of the different kinds of weights in \lpmlnend. 
Roughly speaking, a stable model $I$ of an \lpmln program should satisfy hard rules as many as possible. 
If not, it is a non-probabilistic stable model, which means the knowledge provided by the stable model is not credible. 
For soft rules satisfied by the stable model $I$, a rule with positive weight increases the certainty degree of $I$, 
a rule with negative weight decreases the certainty degree of $I$, 
and a rule with zero weight does not affect the certainty degree of $I$, 
which can be observed from Example \ref{ex:lpmln-inference}.

\subsection{Strong Equivalence for ASP}
Two ASP programs $\Pi_1$ and $\Pi_2$ are strongly equivalent, denoted by $\Pi_1 \equiv_s \Pi_2$, 
if for any ASP program $\Pi_3$, the extended programs $\Pi_1 \cup \Pi_3$ and $\Pi_2 \cup \Pi_3$ have the same stable models \cite{Lifschitz2001Strongly}. 
As presented in \cite{Turner2001SE}, the notion of \textit{strong equivalence models} (SE-models) can be used to characterize the strong equivalence for ASP programs, 
which is defined as follows.  
\begin{defi}[SE-interpretation]
A strong equivalence interpretation (SE-interpretation) is a pair $(X, Y)$ of interpretations such that $X \subseteq Y$. 
An SE-interpretation $(X,Y)$ is called \textit{total} if $X = Y$, and \textit{non-total} if $X \subset Y$.
\end{defi}

\begin{defi}[SE-model for ASP]
\label{def:se-models}
For an ASP program $\Pi$, an SE-interpretation $(X,Y)$ is an SE-model of $\Pi$ if $X \models \Pi^Y$ and $Y \models \Pi$. 
\end{defi} 

By $\ase{\Pi}$, we denote the sets of all SE-models of an ASP program $\Pi$.
Theorem \ref{thm:se-asp} provides a characterization for strong equivalence between ASP programs, that is, the SE-model approach in ASP. 

\begin{thmC}[{\cite[Theorem 1]{Turner2001SE}}]
\label{thm:se-asp}
Two ASP programs $\Pi_1$ and $\Pi_2$ are strongly equivalent iff they have the same SE-models, i.e. $\ase{\Pi_1} = \ase{\Pi_2}$.
\end{thmC}

\begin{exa}
\label{ex:asp-se}
Consider following ASP programs $\Pi_1$
\begin{eqnarray}
	& a \vee b. & \\
	& \leftarrow a, b. &
	\end{eqnarray}
	and $\Pi_2$
	\begin{eqnarray}
	& a \leftarrow  not ~b. & \\
	& b \leftarrow not ~ a. & \\
	& \leftarrow a, b. &
	\end{eqnarray}
It is well-known that $\Pi_1$ and $\Pi_2$ are strongly equivalent ASP programs. 
By Definition \ref{def:se-models}, one can check that $\Pi_1$ and $\Pi_2$ have the same SE-models, i.e. $(\{a\}, \{a\})$ and $(\{b\}, \{b\})$.
\end{exa}

\section{Probabilistic Strong Equivalence}
\label{sec:p-se}
In this section, we investigate probabilistic strong equivalence (p-strong equivalence) between \lpmln programs. 
Firstly, we present several main concepts of equivalences for \lpmln programs including p-ordinary, p-strong, and semi-strong equivalences.
Secondly, we present a model-theoretical characterization for the semi-strong equivalence. 
Thirdly, we present a characterization for the p-strong equivalence based on the characterization of semi-strong equivalence.
Fourthly, we attempt to relax the strict conditions of p-strong equivalence by discussing the notion of q-strong equivalence. 
Finally, we present a formal comparison between the notions and characterizations of strong equivalences presented in this section and Lee and Luo's work \cite{Lee2019LPMLNSE}.

\subsection{Concepts}
Usually, the notion of strong equivalence is built on the notion of ordinary equivalence. 
Intuitively, two \lpmln programs are ordinarily equivalent if their inference results coincide on the MAP and MPD inference tasks, 
which are the most frequently used inference tasks for \lpmlnend.
Therefore, we define the notion of p-ordinary equivalence as follows. 

\begin{defi}[p-ordinary equivalence]
\label{def:lpmln-ordinary-equivalence-p}
Two \lpmln programs $P$ and $Q$ are p-ordinarily equivalent, denoted by $P \equiv_p Q$, 
if $\lsm{P} = \lsm{Q}$ and for any stable model $X \in \lsm{P}$,  we have $Pr(P, X) = Pr(Q, X)$. 
\end{defi}

By the definitions of MAP and MPD inference tasks, it is clear that two p-ordinarily equivalent \lpmln programs have the same MAP and MPD inference results.
Based on the notion of p-ordinary equivalence, we define the notion of p-strong equivalence as follows.  

\begin{defi}[p-strong equivalence]
\label{def:lpmln-strong-equivalence-p}
Two \lpmln programs $P$ and $Q$ are p-strongly equivalent, denoted by $P \equiv_{s,p} Q$, 
if for any \lpmln program $R$, we have $P \cup R \equiv_{p} Q \cup R$.
\end{defi}

According to the above definitions, the p-strong equivalence implies the p-ordinary equivalence, but the inverse does not hold in general, 
which can be observed from the following example.

\begin{exa}
	\label{ex:inference-task-case}
	Consider two \lpmln programs $P$ 
	\begin{eqnarray}
	\alpha &:& a \vee b. \\
	2 &:& \leftarrow a, b. \label{con:1}
	\end{eqnarray}
	and $Q$
	\begin{eqnarray}
	\alpha &:& a \leftarrow  not ~b. \\
	\alpha &:& b \leftarrow not ~ a. \\
	2 &:& \leftarrow a, b. \label{con:2}
	\end{eqnarray}	
	It is easy to check that $P$ and $Q$ have the same stable models and the same probability distribution of stable models, 
	which is shown in Table \ref{tab:computing-results-of-ex1}. 
	Therefore, $P$ and $Q$ are p-ordinarily equivalent, which means they have the same MAP and MPD inference results. 
	Specifically, both of interpretations $\{a\}$ and $\{b\}$ are the most probable stable models of $P$ and $Q$, 
	and for the literals $a$ and $b$, we have $Pr(P, a) = Pr(Q, a) =0.5$ and $ Pr(P, b) = Pr(Q, b) =0.5$.
	But the programs $P$ and $Q$ are not p-strongly equivalent. 
	For example, consider an \lpmln program $R = \{1 : a \leftarrow b.  ~1 : b  \leftarrow a. \}$, 
	the interpretation $S = \{a, b\}$ is a stable model of $P \cup R$, but $S$ is  not a stable model of $Q \cup R$. 
	Since the reduct $\lglred{(Q \cup R)}{S} = \{a \leftarrow b. ~ b \leftarrow a. \}$, obviously, 
	the interpretation $S$ is not a stable model of $\lglred{(Q \cup R)}{S}$, 
	which means $S \not\in \lsm{Q \cup R}$. 
	Therefore, the extended programs $P \cup R$ and $Q \cup R$ do not have the same stable models, i.e. they are not p-strongly equivalent. 

	\begin{table}
		\caption{Stable Models of the Programs $P$ and $Q$ in Example \ref{ex:inference-task-case}}
		\label{tab:computing-results-of-ex1}
    \def\arraystretch{1.3} 
	\begin{tabular}{c c c c c}
		\hline
		Stable Model $S$ & $W(P,S)$ & $Pr(P,S)$   & $W(Q,S)$  & $Pr(Q,S)$ \\ 
		\hline 
		$\{a\}$ & $e^{\alpha + 2}$ &  $0.5$ & $e^{2\alpha + 2}$ & $0.5$ \\ 
		$\{b\}$ & $e^{\alpha + 2}$  & $0.5$ & $e^{2\alpha + 2}$ & $0.5$ \\ 
		$\emptyset$ & $e^2$ & $0$ & $e^2$ &  $0$ \\ 
		\hline 
	\end{tabular} 
	\end{table}
\end{exa}

Now, we introduce the notion of semi-strong equivalence, 
which helps us describe the characterization of the p-strong equivalence more conveniently. 

\begin{defi}[semi-strong equivalence]
	\label{def:lse-semi}
	Two \lpmln programs $P$ and $Q$ are semi-strongly equivalent, denoted by $P \equiv_{s,s} Q$, 
	if for any \lpmln program $R$, we have $\lsm{P \cup R} = \lsm{Q \cup R}$.
\end{defi}

Recall Example \ref{ex:inference-task-case}, it is easy to observe that the programs $P$ and $Q$ in Example \ref{ex:inference-task-case} are not semi-strongly equivalent. 
Definition \ref{def:lse-semi} relaxes the p-strong equivalence by ignoring the probability distribution of stable models. 
Therefore, for \lpmln programs $P$ and $Q$, $P \equiv_{s, p} Q$ implies $P \equiv_{s, s} Q$, but the inverse does not hold in general. 
Based on the relationship between two notions of strong equivalences for \lpmlnend, 
the characterization of the p-strong equivalence can be divided into two parts: 
(1) characterizing the semi-strong equivalence by introducing the notion of SE-models for \lpmlnend; 
(2) characterizing the p-strong equivalence by introducing uncertainty measurement conditions on the basis of the semi-strong equivalence, 
which are shown in the following subsections. 

\subsection{Characterizing Semi-Strong Equivalence}
Similar to the SE-model approach for characterizing the strong equivalence of ASP, we introduce the SE-model for \lpmlnend, 
which can be used to characterize the semi-strong equivalence between \lpmln programs. 

\begin{defi}[SE-model for \lpmlnend]
	\label{def:lse-model}
	For an \lpmln program $P$, an SE-interpretation $(X,Y)$ is an SE-model of $P$, if $X \models \lglred{P}{Y}$.
\end{defi}

In Definition \ref{def:lse-model}, $\lglred{P}{Y}$ is an ASP program obtained from $P$ by a three-step transformation of the program $P$. 
In the first step, $P_Y$ is the \lpmln reduct of $P$ w.r.t. $Y$. 
In the second step, $\overline{P_Y}$ is the unweighted ASP counterpart of $P_Y$.
In the third step, $(\overline{P_Y})^Y$ is  the GL-reduct of $\overline{P_Y}$ w.r.t. $Y$. 
Through the transformation, an SE-model for \lpmln is reduced to an SE-model for ASP, which shows the relationship between \lpmln and ASP. 
By $\lse{P}$, we denote the set of all SE-models of an \lpmln program $P$. 

\begin{defi}
	\label{def:lse-weight-degree}
	For an \lpmln program $P$ and an SE-model $(X,Y)$ of $P$, the weight degree $W(P,(X,Y))$ of $(X,Y)$ w.r.t. the program $P$ is defined as 
	\begin{equation}
	\label{eq:lse-model-weight-degree}
	W(P,(X,Y)) = W(P_Y) = exp \left( \sum_{w:r \in P_Y} w \right)
	\end{equation}
\end{defi}

\begin{exa}
	\label{ex:se-model}
	Continue Example \ref{ex:inference-task-case}, for the \lpmln programs $P$ and $Q$ in the example, 
	consider an interpretation $U = \{a, b\}$, 
	we have $\overline{P_U} = \{a \vee b. \}$ and $\overline{Q_U} = \{a\leftarrow not ~b. ~ b\leftarrow not ~a.\}$. 
	For the GL-reduct, we have $\lglred{P}{U} = \{a \vee b. \}$  and  $\lglred{Q}{U} = \emptyset$. 
	Obviously, for any subset $U'$ of $U$, SE-interpretation $(U', U)$ is an SE-model of $Q$, but $(\emptyset, U)$ is not an SE-model of $P$. 
	All SE-models of $P$ and $Q$ are shown in Table \ref{tab:se-models}. 
	
	\begin{table}
		\begin{threeparttable}
			\centering
			\caption{SE-models and Their Weights of Programs in Example \ref{ex:se-model}}
			\label{tab:se-models}
    \def\arraystretch{1.3} 
			\begin{tabular}{cccccccc}
				\hline
				SE-model $S$ & $(\emptyset, \emptyset)$ & $(\{a\}, \{a\})$ & $(\{b\}, \{b\})$ & $(\emptyset, U)$ & $(\{a\}, U)$ & $(\{b\}, U)$ & $(U, U)$ \\ 
				\hline
				$W(P, S)$ & $e^{2}$ & $e^{\alpha + 2}$ & $e^{\alpha + 2}$ & $-$ & $e^{\alpha}$ & $e^{\alpha}$ & $e^{\alpha}$ \\ 
				$W(Q, S)$ & $e^{2}$ & $e^{2\alpha + 2}$ & $e^{2\alpha + 2}$ & $e^{2\alpha}$ & $e^{2\alpha}$ & $e^{2\alpha}$ & $e^{2\alpha}$ \\ 
				\hline
			\end{tabular}
			\begin{tablenotes}
				\item[*] $U = \{a, b\}$, and ``$-$'' means ``not an SE-model''.
			\end{tablenotes}
		\end{threeparttable}
	\end{table}
\end{exa}

Now, we show some properties of the SE-models for \lpmlnend, which will be used to characterize the semi-strong equivalence for \lpmln programs. 
Lemma \ref{lem:XX-lse-model-LM} \textemdash \ref{lem:x-equilibrium-sm} show some immediate results derived from the definition of SE-models for \lpmlnend, 
which can also be observed from Example \ref{ex:se-model}. 
\begin{lem}
\label{lem:XX-lse-model-LM}
For an \lpmln program $P$, an arbitrary total SE-interpretation $(X, X)$ is an SE-model of $P$, i.e. $(X, X) \in \lse{P}$.
\end{lem}

\begin{lem}
\label{lem:XY-not-lse-model}
For an \lpmln program $P$,  an SE-interpretation $(X, Y)$ is not an SE-model of $P$ iff $X \not\models \lglred{P}{Y}$.
\end{lem}

\begin{lem}
\label{lem:x-equilibrium-sm}
For an \lpmln program $P$ and an interpretation $X$, 
$X$ is a stable model of $P$, iff  $(X',X)$ is not an SE-model of $P$  for any proper subset $X'$ of $X$.
\end{lem}

Based on above properties of SE-models, Lemma \ref{lem:lpmln-sm-strong-equiv} provides a model-theoretical characterization for the semi-strong equivalence. 

\begin{lem}
	\label{lem:lpmln-sm-strong-equiv}
	Two \lpmln programs $P$ and $Q$  are semi-strongly equivalent, 
	iff they have the same SE-models, i.e. $\lse{P} = \lse{Q}$.
\end{lem}

\begin{proof}
For the if direction, 
suppose $\lse{P} = \lse{Q}$, we need to prove that for any \lpmln program $R$, the programs $P \cup R$ and $Q \cup R$ have the same stable models. 
We use proof by contradiction. 
For an interpretation $Y$, assume that $Y \in \lsm{P \cup R}$ and  $Y \not\in \lsm{Q \cup R}$.
By the definition of stable model, we have $Y \models \lglred{(Q \cup R)}{Y}$, and there is a proper subset $X$ of $Y$ such that $X \models \lglred{(Q \cup R)}{Y}$, which means $X \models \lglred{Q}{Y}$ and $X \models \lglred{R}{Y}$.
By the definition of SE-model, we have $(X, Y)$ is an SE-model of $Q$. 
Since $\lse{P} = \lse{Q}$, $(X,Y)$ is also an SE-model of $P$, which means $X \models \lglred{P}{Y}$. 
Combining the above results, we have $X \models \lglred{(P \cup R)}{Y}$, 
which contradicts with $Y \in \lsm{P \cup R}$ by Lemma \ref{lem:x-equilibrium-sm}. 
Therefore, the programs $P \cup R$ and $Q \cup R$ have the same stable models, and the if direction of Lemma \ref{lem:lpmln-sm-strong-equiv} is proven.

For the only-if direction, suppose $\lsm{P \cup R} = \lsm{Q \cup R}$, we need to prove that $\lse{P} = \lse{Q}$. 
We use proof by contradiction. 
For an SE-interpretation $(X,Y)$, assume that $(X,Y) \in \lse{P}$ and $(X, Y) \not\in \lse{Q}$, 
we have $X \not\models \lglred{Q}{Y}$ by Lemma \ref{lem:XY-not-lse-model}. 
Now we show that there is an \lpmln program $R$ such that $\lsm{P \cup R} \neq \lsm{Q \cup R}$ under the assumption. 
Let $R = \{1 : a. ~|~ a \in X\} \cup \{1 : a \leftarrow b. ~|~ a, b \in Y-X \}$,  
we have $\lglred{(Q \cup R)}{Y} = \lglred{Q}{Y} \cup \overline{R}$ and $X \models R$. 
Let $X'$ be a set of literals such that $X' \subseteq Y$ and $X' \models \lglred{Q}{Y} \cup \overline{R}$. 
By the construction of $R$, we have $X \subseteq X'$. 
Since $X \not\models \lglred{Q}{Y}$, we have $X \neq X'$, 
which means there at least exists a literal $l \in Y-X$ such that $l \in X'$.
By the construction of $R$, we have $X' \models 
\overline{R}$ iff $(Y-X) \subseteq X'$, which means $X' = Y$.
By the definition of stable models, $Y$ is a stable model of $Q \cup R$. 
Since $(X, Y) \in \lse{P}$, we have $X \models \lglred{P}{Y}$ and $X \models \lglred{(P \cup R)}{Y}$, which means $Y \not\in \lsm{P \cup R}$. 
Therefore, $P$ and $Q$ have the same SE-models, and the only-if direction of Lemma \ref{lem:lpmln-sm-strong-equiv} is proven. 
\end{proof}
\begin{exa}
	\label{ex:semi-strong-equivalence}
	Continue Example \ref{ex:se-model}, it is easy to check that programs $P$ and $Q$ in Example \ref{ex:se-model} are not semi-strongly equivalent by Lemma \ref{lem:lpmln-sm-strong-equiv}. 
	Next, consider new \lpmln programs $P'$
	\begin{eqnarray}
	w_1 &:& a \vee b. \\
	w_2 &:& b \leftarrow a. 
	\end{eqnarray}
	and $Q'$
	\begin{eqnarray}
	w_3 &:& b. \\
	w_4 &:& a  \leftarrow not ~b. 
	\end{eqnarray}
	where $w_i ~(1 \leq i \leq 4)$ is a variable denoting the weight of corresponding rule. 
	All SE-models of the new \lpmln programs and their weight degrees are shown in Table \ref{tab:semi-strong-equivalence}.
	It is easy to check that $(\{b\}, \{a, b\})$ is the unique non-total SE-model of $P'$ and $Q'$, therefore, the \lpmln programs $P'$ and $Q'$ are semi-strongly equivalent.

	\begin{table}
		\caption{Computing Results in Example \ref{ex:semi-strong-equivalence}}
		\label{tab:semi-strong-equivalence}
    \def\arraystretch{1.3} 
		\begin{tabular}{cccccc}
			\hline
			SE-model $S$ & $(\emptyset, \emptyset)$ & $(\{a\}, \{a\})$ & $(\{b\}, \{b\})$ & $(\{b\}, \{a, b\})$ & $(\{a,b\}, \{a, b\})$ \\ 
			\hline
			$W(P', S)$ & $e^{w_2}$ & $e^{w_1}$ & $e^{w_1 + w_2}$ & $e^{w_1 + w_2}$ & $e^{w_1 + w_2}$ \\ 
			$W(Q', S)$ & $e^{0}$ & $e^{w_4}$ & $e^{w_3 + w_4}$ & $e^{w_3 + w_4}$ & $e^{w_3 + w_4}$ \\ 
			\hline
		\end{tabular}
	\end{table}
\end{exa}

\subsection{Characterizing P-Strong Equivalence}
Now, we investigate the characterization of the p-strong equivalence for \lpmln programs. 
Due to the hard rules of \lpmlnend, the \textit{lim} operation is used in computing the probability degree of a stable model, 
which makes the characterization of p-strong equivalence complicated. 
Therefore, we firstly present a sufficient condition for characterizing the notion. Then we investigate whether the condition is necessary.

Lemma \ref{lem:lpmln-strong-equivalence-suf} provides a sufficient condition for characterizing the p-strong equivalence between \lpmln programs, called \textit{PSE-condition}, 
which adds new conditions w.r.t. the weights of SE-models on the basis of semi-strong equivalence.

\begin{lem}
	\label{lem:lpmln-strong-equivalence-suf}
	Two \lpmln programs $P$ and $Q$ are p-strongly equivalent, if they are semi-strongly equivalent, 
	and there exist two constants $c$ and $k$ such that 
	for each SE-model $(X,Y) \in \lse{P}$, we have $W(P,(X,Y)) = exp(c + k*\alpha) * W(Q,(X,Y))$. 
\end{lem}

\begin{proof}
For \lpmln programs $P$ and $Q$, to show the p-strong equivalence between them, 
we need to show that, for any \lpmln program $R$, 
\begin{enumerate}
	\item $P \cup R$ and $Q \cup R$ have the same stable models, i.e. $\lsm{P \cup R} = \lsm{Q \cup R}$; 
	\item $P \cup R$ and $Q \cup R$ have the same probabilistic stable models, i.e. $\psm{P \cup R} = \psm{Q \cup R}$; and 
	\item  $P \cup R$ and $Q \cup R$ have the same probability distribution of their stable models, i.e. for any stable model $X \in \psm{P \cup R}$, $Pr(P \cup R, X) = Pr(Q \cup R, X)$.
\end{enumerate}

\textbf{Part 1.} Since the programs $P$ and $Q$ are semi-strongly equivalent,
by Lemma \ref{lem:lpmln-sm-strong-equiv}, we have $\lsm{P \cup R} = \lsm{Q \cup R}$ and $\lse{P} = \lse{Q}$. 

\textbf{Part 2.} We prove $\psm{P \cup R} = \psm{Q \cup R}$ by contradiction. Without loss of generality, assume $X$ is an interpretation such that $X \in \psm{P \cup R}$ and $X \not\in \psm{Q \cup R}$, 
by the PSE-condition, we have 
\begin{equation}
	h(P \cup R, X) = h(P, X) + h(R, X) = h(Q, X) + h(R, X) + k = h(Q \cup R, X) + k
\end{equation}
And for any stable model $X' \in \lsm{P \cup R}$, we have $h(P \cup R, X') \leq h(P \cup R, X)$. 
Suppose $Y$ is a probabilistic stable model of $Q \cup R$. 
Since $X \not\in \psm{Q \cup R}$, we have $h(Q \cup R, Y) > h(Q \cup R, X)$,  
which means $h(P \cup R, Y) > h(P \cup R, X)$. 
It contradicts with the assumption, 
therefore, $\psm{P \cup R} = \psm{Q \cup R}$, which means for any non-probabilistic stable model $X \in \lsm{P \cup R} - \psm{P \cup R}$, $Pr(P \cup R, X) = Pr(Q \cup R, X) = 0$. 

\textbf{Part 3.} For a stable model $X \in \psm{P \cup R}$, the probability degree of $X$ can be reformulated as 
\begin{equation}
\label{eq:p-se-probability-equal}
\begin{split}
Pr(P \cup R, X)
& = \frac{W(P \cup R,X)}{\Sigma_{X'\in \psm{P \cup R}}W(P \cup R,X')} \\
& = \frac{W(P,X) * W(R, X)}{\Sigma_{X'\in \psm{P \cup R}}W(P,X') * W(R, X')} \\
& = \frac{ exp(c + k*\alpha) * W(Q,X) * W(R, X)}{exp(c + k*\alpha) * \Sigma_{X'\in \psm{Q \cup R}}  W(Q, X') * W(R, X')} \\
& = \frac{W(Q \cup R,X)}{\Sigma_{X'\in \psm{Q \cup R}} W(Q \cup R,X')} 
= Pr(Q \cup R, X)
\end{split} 
\end{equation}

Combining above results, $P$ and $Q$ are p-strongly equivalent, Lemma \ref{lem:lpmln-strong-equivalence-suf} is proven. 
\end{proof}
Lemma \ref{lem:lpmln-strong-equivalence-suf} shows that the PSE-condition is a sufficient condition for characterizing the p-strong equivalence.  
One may ask whether the PSE-condition is also necessary. 
Fortunately, the answer is yes, but it is not easy to prove due to the hard rules in \lpmlnend. 
Recall p-ordinarily equivalent programs $P$ and $Q$ in Example \ref{ex:inference-task-case}, 
it is easy to check that $e^\alpha * W(P, \{a\}) = W(Q, \{a\})$ and $e^\alpha * W(P, \{b\}) = W(Q, \{b\})$, 
while $e^\alpha * W(P, \emptyset) \neq W(Q, \emptyset)$. 
Although the example is not for p-strongly equivalent programs, 
it still shows how hard rules affect the characterization of p-strong equivalence. 
That is, for p-strongly equivalent \lpmln programs $P$ and  $Q$, if there is an interpretation $X$ such that for any \lpmln program $R$, $X \not\in \psm{P \cup R}$, 
the condition on the weight of SE-models $(X', X)$ is not necessary for characterizing the p-strong equivalence between $P$ and $Q$. 
Since if $X$ cannot be a probabilistic stable model, 
its probability degree is zero, 
which means $Pr(P \cup R, X)$ is always equal to $Pr(Q \cup R, X)$. 
Above intuition is formally described by following lemmas, 
which serve as some preconditions of proving the necessity of the PSE-condition. 

First of all, we introduce some notations. 
For a set $U$ of literals, we use $2^U$ to denote the power set of $U$, 
and use $2^{U^+}$ to denote the set of interpretations in $2^U$, which is 
\begin{equation}
2^{U^+} = \{X \in 2^U~|~ X \text{ is consistent }\}
\end{equation}
For an \lpmln program $P$, a set $E$ of \lpmln programs is called a set of necessary extensions w.r.t. $P$, if for any interpretations $X$ and $Y$, 
there exists a program $P' \in E$ such that $P \subseteq P'$, and both of $X$ and $Y$ are probabilistic stable models of $P'$. 

Lemma \ref{lem:pse-onlyif-part1} shows that for any p-strongly equivalent \lpmln programs $P$ and $Q$, 
if there exists a set of necessary extensions w.r.t. $P$ and $Q$, then the PSE-condition w.r.t. $P$ and $Q$ is necessary, 
which means to prove the necessity of the PSE-condition, we need to construct a set of necessary extensions.

\begin{lem}
	\label{lem:pse-onlyif-part1}
	For p-strongly equivalent \lpmln programs $P$ and $Q$, let $R_1$ and $R_2$ be arbitrary \lpmln programs such that $\psm{Q \cup R_1} \cap \psm{Q \cup R_2} \neq \emptyset$. 
	There exist two constants $c$ and $k$ such that for any SE-models $(X,Y)$ of $Q$, 
	if $Y \in \psm{Q \cup R_1} \cup \psm{Q \cup R_2}$, then $W(P,(X,Y)) = exp(c + k*\alpha) * W(Q,(X,Y))$.
\end{lem}

\begin{proof}
Suppose $I$ is a probabilistic stable model such that $I \in \psm{Q \cup R_1} \cap \psm{Q \cup R_2}$, 
and there are two constants $c$ and $k$ such that $W(P, I) = exp(c + k*\alpha) * W(Q, I)$. 
Since the programs $P$ and $Q$ are p-strongly equivalent, for any probabilistic stable model $X \in \psm{Q \cup R_1}$, 
we have $Pr(P \cup R_1, X) = Pr(Q \cup R_1, X)$. 
By $SW(Q)$, we denote the sum of weight degrees of probabilistic stable models of an \lpmln program $Q$, i.e. $SW(Q) = \sum_{X \in \psm{Q}} W(Q, X)$. 
By the definition of probability degree, for any probabilistic stable model $X \in \psm{Q \cup R_1}$,  
we have 
\begin{equation}
    Pr(P \cup R_1, X) = Pr(Q \cup R_1, X) = \frac{W(P, X) \times W(R_1, X)}{SW(P \cup R_1)}  =\frac{ W(Q, X)  \times W(R_1, X) }{SW(Q \cup R_1)}
\end{equation}
therefore, we have $W(P, X) = (SW(P \cup R_1) / SW(Q \cup R_1)) * W(Q, X)$ for any $X \in \psm{Q \cup R_1}$. 
Since $I \in \psm{Q \cup R_1}$, we have $SW(P \cup R_1) / SW(Q \cup R_1) = exp(c + k * \alpha)$. 
Similarly, for any probabilistic stable model $X' \in \psm{Q \cup R_2}$, we can derive that  $W(P, X') = exp(c + k * \alpha) * W(Q, X')$. 
Therefore, for any interpretation $Y \in \psm{Q \cup R_1} \cup \psm{Q \cup R_2}$, 
we have shown that $W(P, Y) = exp(c + k * \alpha) * W(Q, Y)$. 
By the definition of the weight degree of SE-model, Lemma \ref{lem:pse-onlyif-part1}  is proven. 
\end{proof}

Proposition \ref{prop:se-model-universe} shows that to construct the necessary extensions w.r.t. an \lpmln program $P$, we only need to consider the interpretations consisting of literals occurred in $P$.

\begin{prop}
\label{prop:se-model-universe}
An SE-interpretation $(X, Y)$ is an SE-model of an \lpmln program $P$ iff $(X \cap \olit{P}, Y \cap \olit{P})$ is an SE-model of $P$, and $W(P, Y) = W(P, Y \cap \olit{P})$.
\end{prop}

The proof of Proposition \ref{prop:se-model-universe} is straightforward by the definition of SE-models, 
therefore, we omit the details for brevity. 
By Proposition \ref{prop:se-model-universe}, a set $E$ of \lpmln programs is a set of necessary extensions w.r.t. an \lpmln programs $P$,  
if for any interpretations $X$ and $Y$ in $2^{U^+}$, there exists a program $P'$ of $E$ such that $P \subseteq P'$ and  both of $X$ and $Y$ are probabilistic stable models of $P'$, 
where $\olit{P} \subseteq U$. 
In what follows, we present a method to construct the necessary extensions w.r.t. an \lpmln program. 
Definition \ref{def:r-mxa} provides a basic unit to build a necessary extension, i.e. the flattening rules, 
and Lemma \ref{lem:flattening-rules} shows some important properties of flattening rules. 

\begin{defi}[flattening rules]
	\label{def:r-mxa}
	For two interpretations $X$ and $Y$ such that $X \cap Y = \emptyset$, and an atom $a$ such that neither $a$ or its negation $\neg a$ does not occur in $X \cup Y$, 
	by $R(X,Y,a)$ we denote an \lpmln program as follows
	\begin{eqnarray}
	\alpha &:& \leftarrow X, ~not ~ Y, ~a. \label{def:mxa-rule-1}\\
	\alpha &:& a \leftarrow X, ~not ~ Y. \label{def:mxa-rule-2}
	\end{eqnarray}
\end{defi}

\begin{lem}
\label{lem:flattening-rules}
Let $U$ be a set of literals, $X$ an interpretation of $2^{U^+}$, and  $R(X, U-X, a')$ the flattening rules w.r.t. $X$, $U-X$ and $a'$. 
For any interpretation $I$, 
\begin{itemize}
	\item if $I \cap U = X$, $I$ only satisfies one of the rules in $R(X, U-X, a')$; 
	\item if $I \cap U \neq X$, $I$ satisfies all rules in $R(X, U-X, a')$. 
\end{itemize}
\end{lem}

\begin{proof}
\textbf{Part 1. }
For the case $I \cap U = X$, we have $X \subseteq I$ and $(U-X) \cap I = \emptyset$. 
For flattening rules $R(X, U-X, a')$, if $a' \in I$, $I$ only satisfies the rule of the form \eqref{def:mxa-rule-2}; and if  $a' \not\in I$, $I$ only satisfies the rule of the form \eqref{def:mxa-rule-1}. 
Therefore, $I$ only satisfies one of the rules in $R(X, U-X, a')$. 

\textbf{Part 2. }
For the case $I \cap U \neq X$, there are two cases: $I \cap U \subset X$ and $I \cap U \not\subset X$. 
If $I \cap U \subset X$, it is easy to check the positive bodies of two rules in $R(X, U-X, a')$ cannot be satisfied by $I$, which means $I \models R(X, U-X, a')$. 
If $I \cap U \not\subset X$, it is easy to check the negative bodies of two rules in $R(X, U-X, a')$ cannot be satisfied by $I$, which means $I \models R(X, U-X, a')$. 
\end{proof}
\begin{exa}
Consider a set $U = \{a, b, c\}$ of literals and an interpretation $X=\{a, b\}$, 
by Definition \ref{def:r-mxa}, the flattening rules $R(X, U-X, a')$ are as follows
\begin{eqnarray}
\alpha &:& \leftarrow a, b, not ~c, a'. \label{ex:mxa-rule-1} \\
\alpha &:& a' \leftarrow a, b, not ~c. \label{ex:mxa-rule-2}
\end{eqnarray}
It is easy to check that $X$ only satisfies rule \eqref{ex:mxa-rule-1}, 
and other consistent subsets of $U$ satisfy total $R(X, U-X, a')$. 
Therefore, the flattening rules $R(X, U-X, a')$ can be used to relatively decrease the number of hard rules satisfied by $X$. 
\end{exa}

Lemma \ref{lem:flattening-rules} shows that for a set $U$ of literals and an interpretation $X \in 2^{U^+}$, the flattening rules $R(X, U-X, a')$ can be used to narrow the difference between the numbers of hard rules satisfied by $X$ and other interpretations. 
Moreover, for other interpretations, $R(X, U-X, a')$ does not change the difference among the numbers of hard rules satisfied by them. 
Therefore, the flattening rules can be used to adjust the probabilistic stable models of \lpmln programs, 
which is shown in Definition \ref{def:flattening-extension} and Lemma \ref{lem:flattening-extension-prop}.

\begin{defi}[flattening extension]
	\label{def:flattening-extension}
	For an \lpmln program $P$ and a set $U$ of literals such that $lit(\overline{P}) \subseteq U$, 
	a flattening extension $E^k(P, U)$ of $P$ w.r.t. $U$ is defined as 
	\begin{itemize}
		\item $E^0(P, U) = P \cup R_0$, where $R_0$ is a set of weighted facts constructed from $U$, which is 
		\begin{equation}
			R_0 = \{\alpha : a_k. ~|~ a_k \in U\}
		\end{equation}
		\item $E^{i+1}(P, U) = E^i(P, U) \cup R(X \cap U, U-X,c_{i+1})$, where $X \in \psm{E^i(P, U)}$ and $c_{i+1}$ is an atom such that $c_{i+1} \not\in  \oat{E^i(P, U)}$. 
	\end{itemize}
\end{defi}

Note that in the second step of Definition \ref{def:flattening-extension}, $X$ could be an arbitrary probabilistic stable model of the current program. 
Since the flattening rules are used to narrow the difference between the numbers of hard rules satisfied by probabilistic and non-probabilistic stable models, 
and all of the probabilistic stable models satisfy the same numbers of hard rules. 
But if a specific interpretation $Y$ is required to be a probabilistic stable model of a flattening extension, we should avoid picking $Y$ to construct new flattening extensions. 
Since $Y$ may not be a probabilistic stable model of a flattening extension w.r.t. $Y$, which can be observed from the following results. 

\begin{lem}
	\label{lem:flattening-extension-prop}
	For an \lpmln program $P$ and a set $U$ of literals such that $lit(\overline{P}) \subseteq U$, 
	$E^{k+1}(P, U)$ is a flattening extension of $P$ in Equation \eqref{eq:ex-flattening-extension-program}, 
	\begin{equation}
		\label{eq:ex-flattening-extension-program}
		E^{k+1}(P, U) = E^k(P, U) \cup R(X \cap U, U-X,c_{k+1})
	\end{equation}
	we have following results 
	\begin{itemize}
		\item $\lsm{E^0(P, U)} = 2^{U^+}$; 
		\item $\lsm{E^{k+1}(P, U)} = \lsm{E^k(P, U)} \cup  
		\{ Y \cup \{c_{k+1}\} ~|~ Y \in \lsm{E^k(P, U)} \text{ and } Y \cap U = X \cap U\} \}$; and 
		\item the weight degrees of stable models have following relationships
		\begin{equation}
		W(E^{k+1}(P, U), Y) = \begin{cases}
		W(E^{k}(P, U), Y) * e^{2\alpha} & \text{ if } Y \cap U \neq X \cap U, \\
		W(E^{k}(P, U), Y) * e^{\alpha} & \text{ otherwise. }
		\end{cases}
		\end{equation}
		and for two stable models $Y$ and $Z$ of $E^i(P, U)$ $(i > 0)$, if $Y \cap U = Z \cap U$, then $W(E^i(P, U), Y) = W(E^i(P, U), Z)$.
	\end{itemize}
\end{lem}

\begin{exa}
	\label{ex:flattening-extension}
	Recall the \lpmln program $P$ in Example \ref{ex:se-model}, let $U = \{a, b\}$ be a set of literals, 
	it is clear that $\olit{P} = U$. 
	By Definition \ref{def:flattening-extension}, $E^0(P, U) = P \cup \{\alpha : a. ~ \alpha : b.  \}$, 
	it is easy to check that all subsets of $U$ are the stable models of $E^0(P, U)$, and $U$ is the unique probabilistic stable model. 
	By Definition \ref{def:r-mxa}, the flattening rules $R(U, \emptyset, c_{1})$ are as follows 
	\begin{eqnarray}
	\alpha &:& \leftarrow a, ~b, ~c_1.\\
	\alpha &:& c_1 \leftarrow a, ~b.
	\end{eqnarray}
	and we have  $E^1(P, U) = E^0(P, U) \cup R(U, \emptyset, c_1)$. 
	The stable models and their weight degrees of $P$, $E^0(P, U)$, and $E^1(P, U)$ are shown in Table \ref{tab:flattening-extension}.
	From the table, we can observe that $U$ is a probabilistic stable model of $E^0(P, U)$.  
	After adding rules $R(U, \emptyset, c_{1})$,  the number of hard rules satisfied by $U$ decreases relatively. 
	Although $U$ is still a probabilistic stable model of $E^1(P, U)$, non-probabilistic stable models $\{a\}$ and $\{b\}$ of $E^0(P, U)$ become probabilistic stable models of  $E^1(P, U)$. 
	Let $E^2(P, U) = E^1(P, U) \cup R(U, \emptyset, c_2)$, it is easy to check that the interpretation $U$ becomes a non-probabilistic stable model of $E^2(P, U)$.
	\begin{table}
		\begin{threeparttable}
			\centering
			\caption{Computing Results in Example \ref{ex:flattening-extension}}
			\label{tab:flattening-extension}
      \def\arraystretch{1.3} 
			\begin{tabular}{cccccccc}
				\hline
				Stable Model $S$ & $\emptyset$ & $\{a\}$  & $\{b\}$ & $\{a, b\}$ & $\{a, b, c_1\}$ & $\{a, b, c_2\}$ & $\{a, b, c_1, c_2\}$ \\ 
				\hline
				$W(P, S)$ & $e^2$ & $e^{\alpha + 2}$ & $e^{\alpha + 2}$ & $-$ & $-$  & $-$ & $-$ \\ 
				$W(E^0(P, U), S)$ & $e^2$ & $e^{2\alpha + 2}$ & $e^{2\alpha + 2}$ & $e^{3\alpha}$ & $-$  & $-$ & $-$ \\ 
				$W(E^1(P, U), S)$ & $e^{2\alpha + 2}$ & $e^{4\alpha + 2}$ & $e^{4\alpha + 2}$ & $e^{4\alpha}$ & $e^{4\alpha}$  & $-$ & $-$ \\ 
				$W(E^2(P, U), S)$ & $e^{4\alpha + 2}$ & $e^{6\alpha + 2}$ & $e^{6\alpha + 2}$ & $e^{5\alpha}$ & $e^{5\alpha}$ & $e^{5\alpha}$ & $e^{5\alpha}$ \\ 
				\hline
			\end{tabular}
			\begin{tablenotes}
				\item[*] ``$-$'' means ``not a stable model''.
			\end{tablenotes}
		\end{threeparttable}
	\end{table}
\end{exa}

By Lemma \ref{lem:flattening-extension} and Example \ref{ex:flattening-extension}, 
it has shown that how flattening extensions adjust the numbers of hard rules satisfied by a stable model. 
The following results show how to construct a necessary extension of two p-strongly equivalent programs by using flattening extensions. 

\begin{lem}
	\label{lem:flattening-extension}
	Let $P$ and $Q$ be p-strongly equivalent \lpmln programs, and $U = lit(\overline{P \cup Q})$. 
	For any interpretations $X$ and $Y$ of $2^{U^+}$, there exists a flattening extension $E^k(P, U)$ such that both $X$ and $Y$ are probabilistic stable models of $E^k(P, U)$.
\end{lem}

\begin{proof}
By Lemma \ref{lem:flattening-extension-prop}, we have both of $X$ and $Y$ are stable models of the program $E^0(P, U)$.
Without loss of generality, we assume $h(E^0(P, U), X) \geq h(E^0(Q, U), Y)$. 
We prove Lemma \ref{lem:flattening-extension} by showing a two-step method to construct a flattening extension $E^k(P, U)$ such that $X \in \psm{E^k(P, U)}$ and $Y \in \psm{E^k(P, U)}$. 
In Step 1, we show there exists a minimal number $k_1$ such that $X \in \psm{E^{k_1}(P, U)}$; 
in Step 2, we show that there exists a minimal number $k_2 \geq k_1$ such that $X \in \psm{E^{k_2}(P, U)}$ and $Y \in \psm{E^{k_2}(P, U)}$. 
In the proof, by $d(X, i)$, we denote the difference between numbers of hard rules of $E^i(P, U)$ satisfied by a probabilistic stable model and $X$, 
i.e. $d(X, i) = h(E^i(P, U), X') - h(E^i(P, U), X)$, where $X' \in \psm{E^i(P, U)}$. 
According to the definition, it is easy to check that the minimum value of $d(X, i)$ is zero. 

\textbf{Step 1.}
If $d(X, 0) = 0$, then $E^0(P, U)$ is a minimal flattening extension such  that $X$ becomes a probabilistic stable model. 
If $d(X, 0) > 0$ and $n$ is an integer such that $d(X, i) > 0 ~ (0 \leq i \leq n)$, there are two possible cases for each number $i ~(0 \leq i \leq n)$: 
(1) $E^i(P, U)$ has exactly one probabilistic stable model, i.e. $|\psm{E^i(P, U)}| = 1$, it is easy to check that $d(X, i+1) = d(X, i) -1$; 
(2) $E^i(P, U)$ has multiple probabilistic stable models, i.e. $|\psm{E^i(P, U)}| > 1$, it is easy to check that $d(X, i+1) = d(X, i)$ and $|\psm{E^{i+1}(P, U)}| = |\psm{E^i(P, U)}| - 1$. 
Therefore, $d(X, i)$ is a monotonically decreasing function over the interval $(0, n)$, and there is always an integer $j > i$ such that $d(X, j) < d(X, i)$, 
which means there exists a minimal number $k_1 > n$ such that $d(X, k_1) = 0$. 
That is, $X$ becomes a probabilistic stable model of $E^{k_1}(P, U)$.

\textbf{Step 2.} 
Since $h(E^0(P, U), Y) \leq h(E^0(P, U), X)$, 
for the flattening extension $E^{k_1}(P, U)$, we have $h(E^{k_1}(P, U), Y) = h(E^{0}(P, U), Y) + 2  k_1 \leq h(E^{0}(P, U), X) + 2  k_1 = h(E^{k_1}(P, U), X)$. 
If $h(E^{k_1}(P, U), Y) = h(E^{k_1}(P, U), X)$, $E^{k_1}(P, U)$ is the flattening extension such that both of $X$ and $Y$ become probabilistic stable models.
Otherwise, we have $X \in \psm{E^{k_1}(P, U)}$, while $Y \not\in \psm{E^{k_1}(P, U)}$, 
therefore, we need to further extend $E^{k_1}(P, U)$. 
As we discussed in Step 1, there is a minimal number $k_2 \geq k_1$ such that $Y$ becomes a probabilistic stable model of $E^{k_2}(P, U)$.
We need to show that $X$ is also a probabilistic stable model of $E^{k_2}(P, U)$. 
We use proof by contradiction. 
Assume $X \not\in \psm{E^{k_2}(P, U)}$, there exists an integer $k'$ such that $k_1 < k' < k_2$, $h(E^{k'}(P, U), X) = h(E^{k'}(P, U), Y)$ and $E^{k'+1}(P, U) = E^{k'}(P, U) \cup R(X \cap U, U-X, c')$, 
which means both of $X$ and $Y$ are probabilistic stable models of $E^{k'}(P, U)$. 
It contradicts with the premise that $k_2$ is the minimal integer such that $Y \in 
\psm{E^{k_2}(P, U)}$.
Therefore, both $X$ and $Y$ are probabilistic stable models of $E^{k_2}(P, U)$,  Lemma \ref{lem:flattening-extension} is proven.
\end{proof}
Lemma \ref{lem:flattening-extension} shows that one can construct a set of necessary extensions of two p-strongly equivalent \lpmln programs by constructing a set of flattening extensions. 
Combining Lemma \ref{lem:lpmln-strong-equivalence-suf} and  Lemma \ref{lem:pse-onlyif-part1}, we have found a sufficient and necessary condition to characterize the p-strong equivalence for \lpmln programs, which is shown in Theorem \ref{thm:lpmln-strong-equivalence-w}.

\begin{thm}
	\label{thm:lpmln-strong-equivalence-w}
	Two \lpmln programs $P$ and $Q$ are p-strongly equivalent iff they are semi-strongly equivalent, 
	and there exist two constants $c$ and $k$ such that 
	for each SE-model $(X,Y) \in \lse{P}$, $W(P,(X,Y)) = exp(c + k*\alpha) * W(Q,(X,Y))$.
\end{thm}

\begin{exa}
	\label{ex:p-strong-equivalence}
	Recall Example \ref{ex:semi-strong-equivalence}, it has been shown that the programs $P'$ and $Q'$ are semi-strongly equivalent. 
	If the programs are also p-strongly equivalent, we have following system of linear equations, where $\mathcal{C} =  exp(k*\alpha + c)$ and $U = \{a, b\}$.
	
	\begin{equation}
	\left\{
	\begin{array}{l}
	W(P', (\emptyset, \emptyset)) = W(Q', (\emptyset, \emptyset)) * \mathcal{C} \\
	W(P', (\{a\}, \{a\})) = W(Q', (\{a\}, \{a\})) * \mathcal{C} \\
	W(P', (\{b\}, \{b\})) = W(Q', (\{b\}, \{b\})) * \mathcal{C} \\
	W(P', (U, U)) = W(Q', (U, U)) * \mathcal{C}
	\end{array}
	\right.
	\Rightarrow
	\left\{
	\begin{array}{l}
	w_2 = c + k*\alpha \\
	w_1 = w_4 + c + k * \alpha \\
	w_1 + w_2 = w_3 + w_4 + c + k * \alpha
	\end{array}
	\right.
	\end{equation}
	Solve the system of equations, 
	we have $P'$ and $Q'$ are p-strongly equivalent iff $w_2 = w_3 = c + k*\alpha$ and $w_1 = w_4 + c + k*\alpha$. 
	According to the syntax of \lpmln rules, the value of $w_2$ is either a real number or ``$\alpha$'', which means $w_2 = w_3 = c$ or $w_2 = w_3 = \alpha$.
	Therefore, there are two kinds of solutions of the systems of equations: (1) 
	$w_2 = w_3 = c$ and $w_1 = w_4 + c$; (2) $w_1 = w_2 = w_3 = \alpha$ and $w_4=0$. 
\end{exa}

\subsection{Discussion on Approximate Strong Equivalence}
From Example \ref{ex:p-strong-equivalence}, we can observe that for two semi-strongly equivalent \lpmln programs, 
there are usually some very strict conditions to make the programs p-strongly equivalent, 
which is hard to achieve in many cases. 
Therefore, we need to consider a kind of relaxed notion of the p-strong equivalence. 
 
Recall the notions of strong equivalences introduced in this section,  
the p-strong equivalence can be viewed as a kind of exact strong equivalence, 
and the semi-strong equivalence is a kind of approximate strong equivalence.
By contrast, all probabilistic inference results of two p-strongly equivalent \lpmln programs coincide, 
while none of the probabilistic inference results of two semi-strongly equivalent \lpmln  programs coincide in general. 
A possible way to investigate approximate strong equivalence is to find a kind of intermediate notion between the p-strong and semi-strong equivalences, 
i.e. we only consider the equivalence w.r.t. part of probabilistic inference tasks. 
For example, the following definition of q-strong equivalence seems a kind of approximate strong equivalence.

\begin{defi}[q-strong equivalence]
	\label{def:q-strong-equivalence}
	Two \lpmln programs $P$ and $Q$ are q-strongly equivalent, denoted by $P \equiv_{s,q} Q$, 
	if for any \lpmln programs $R$, $\lsm{P \cup R} = \lsm{Q \cup R}$, 
	and for any stable models $X, Y \in \lsm{P \cup R}$,  $W(P \cup R, X) \leq W(P \cup R, Y)$ iff $W(Q \cup R, X) \leq W(Q \cup R, Y)$. 
\end{defi} 

Apparently, the MPD inference results of two q-strongly equivalent \lpmln do not coincide in general, only the MAP inference results coincide. 
The notion of q-strong equivalence is useful for the programs rewriting of many applications. 
On the one hand, in many problems such as qualitative decision-making, we are interested in the optional solutions or the solutions sorted by optimum degrees, which means the exact optimum degree of a solution is not needed. 
On the other hand, in many scenarios, the probability distribution of data is not easy to get, such as personal preferences for something that are used in the recommender systems. 
However, the following result shows that the q-strong equivalence is just a reformulation of the p-strong equivalence.

\begin{thm}
\label{thm:q-se}
Two \lpmln programs are q-strongly equivalent iff they are p-strongly equivalent. 
\end{thm}

\begin{proof}
By the definition of q-strong and p-strong equivalence, the if direction of Theorem \ref{thm:q-se} is obvious. 
For the only-if direction, we use proof by contradiction. 
By $C_i = exp(c_i + k_i * \alpha)$, we denote a weight expression. 
Assume \lpmln programs $P$ and $Q$ are q-strongly equivalent, but they are not p-strongly equivalent, 
we have $\lse{P} = \lse{Q}$ and there exist interpretations $X$ and $Y$ such that $W(P, X) / W(Q, X) \neq W(P, Y) / W(Q, Y)$,  
suppose $W(Q, X) = C_1 * W(P, X)$ and $W(Q, Y) = C_2 * W(P, Y)$.
Let $R$ be an \lpmln program such that $X, Y \in \lsm{P \cup R}$ and $W(P \cup R, X) \leq W(P \cup R, Y)$. 
Let $W(P, X) = C_3 * W(P, Y)$ and $W(R, X) = C_4 * W(R, Y)$, 
we can derive that $C_3 * C_4 \leq 1$. 
Since $P$ and $Q$ are q-strongly equivalent, we have $W(Q \cup R, X) \leq W(Q \cup R, Y)$, which can be reformulated as follows
\begin{eqnarray}
W(Q \cup R, X) &=& C_1 * C_3 * C_4 * W(P \cup R, Y) \\
W(Q \cup R, Y) &=& C_2 * W(P \cup R, Y)
\end{eqnarray}
Therefore, we have $C_1 * C_3 * C_4 / C_2 \leq 1$. 
Since $C_3 * C_4 \leq 1$ and $C_1 \neq C_2$, 
if $C_1 < C_2$, it is obvious $W(Q \cup R, X) \leq W(Q \cup R, Y)$. 
But if $C_1 > C_2$, we need to consider different \lpmln programs $R$, i.e. different values of $C_4$. 
Recall the flattening rules in Definition \ref{def:r-mxa}, 
the weights of the rules can be an arbitrary real number or the symbol ``$\alpha$'' actually.  
Similar to the proof of Lemma \ref{lem:flattening-extension}, we can obtain arbitrary values of $C_4$ by using proper flattening rules. 
In other words, we can construct an \lpmln program $R$ such that $C_3 * C_4 \leq 1$ and $C_1 * C_3 * C_4 / C_2 > 1$, which means the programs $P$ and $Q$ are not q-strongly equivalent.  
It contradicts with the assumption, therefore, the only-if direction is proven. 
\end{proof}
Theorem \ref{thm:q-se} shows that the q-strong equivalence is not an available notion of approximate strong equivalence for \lpmlnend, 
therefore, the investigation of approximate strong equivalence is still an open problem. 
But it can be observed that results presented in this paper such as the flattening rules and extensions would be a foundation of future works. 
In addition, it is worth noting that Lee and Luo pointed out that 
allowing the probabilistic inference results to be slightly different with some error bounds would be a possible way to study the approximate strong equivalence for \lpmln \cite{Lee2019LPMLNSE}.

\subsection{Relating to Lee and Luo's Results}
Besides the semi-strong equivalence and p-strong equivalence presented in this section, 
Lee and Luo also present two kinds of notions of strong equivalences for \lpmln programs, 
i.e. the structural equivalence and Lee and Luo's strong equivalence (LL-strong equivalence for short) \cite{Lee2019LPMLNSE}. 
Here, we present a formal comparison between these strong equivalences notions by their definitions and characterizations, 
the results are shown in Proposition \ref{prop:se-comparison}.
\begin{prop}
\label{prop:se-comparison}
For the existing strong equivalences notions of \lpmlnend, we have 
\begin{itemize}
    \item the structural equivalence and semi-strong equivalence are equivalent to each other, which can be observed from their definitions and characterizations; and 
    \item the LL-strong equivalence and p-strong equivalence are also equivalent to each other, 
    which can be observed from their characterizations. 
\end{itemize}
\end{prop}

In what follows, we present a detailed discussion for the results in Proposition \ref{prop:se-comparison}. 
Firstly, we introduce the notions of structural strong equivalence and LL-strong equivalence. 
For simplicity, Lee and Luo's original results are reformulated by using the terms and notations of this paper, 
which does not affect the following discussion. 
For example, the notion of soft stable models in \cite{Lee2019LPMLNSE} is exactly the notion of stable models for \lpmln presented in Section \ref{sec:preliminaries}. 
\begin{defiC}[{\cite[Definition 4 and 5]{Lee2019LPMLNSE}}]
\label{def:strong-equivalence-ll19}
For \lpmln programs $P$ and $Q$, 
\begin{itemize}
    \item they are structurally equivalent, denoted by $P \equiv_{s, st} Q$, 
    if $\lsm{P \cup R} = \lsm{Q \cup R}$ for any \lpmln program $R$; 
    \item they are LL-strongly equivalent, denoted by $P  \equiv_{s, pr} Q$, 
    if $Pr(P \cup R, X) = Pr(Q \cup R, X)$ for any \lpmln program $R$ and any interpretation $X$. 
\end{itemize}
\end{defiC}

Secondly, we compare the notions of semi-strong equivalence and structural equivalence. 
By the definitions of strong equivalences in Definition \ref{def:lse-semi} and \ref{def:strong-equivalence-ll19}, 
it is easy to observe that $P \equiv_{s, s} Q$ iff $P \equiv_{s, st} Q$ for any \lpmln programs $P$ and $Q$. 
Therefore, the notions of semi-strong equivalence and structural equivalence are essentially the same. 
To characterize the structural equivalence, Lee and Luo present four different kinds of approaches, which are equivalent to each other.   
We introduce the soft logic of Here-and-There (soft HT-logic) approach, which is shown in Definition \ref{def:soft-ht-model-ll19} and Theorem \ref{thm:soft-ht-model-semi-se}.
\begin{defiC}[{\cite[Definition 6]{Lee2019LPMLNSE}}]
\label{def:soft-ht-model-ll19}
An SE-interpretation $(X, Y)$ is a soft HT-model of an \lpmln program $P$, if $(X, Y)$ is an HT-model of the ASP program $\overline{P_Y}$.
\end{defiC}

\begin{thm}[Theorems on Soft Stable Models from {\cite[Page 203]{Lee2019LPMLNSE}}]
\label{thm:soft-ht-model-semi-se}
\lpmln programs $P$ and $Q$ are structurally equivalent iff they have the same set of soft HT-models. 
\end{thm}

For an SE-interpretation $(X, Y)$, 
it has been shown that $(X, Y)$ is an SE-model of an ASP program iff $(X, Y)$ is an HT-model of the program \cite{Turner2001SE}. 
Therefore, it is easy to observe that $(X, Y)$ is a soft HT-model of an \lpmln program iff $(X, Y)$ is an SE-model of the program, 
which means Theorem \ref{thm:soft-ht-model-semi-se}  and Lemma \ref{lem:lpmln-sm-strong-equiv} are equivalent to each other. 
Since the soft HT-logic characterization for structural equivalence is equivalent to other characterizations presented in \cite{Lee2019LPMLNSE}, 
the SE-model approach presented in this paper is also equivalent to those characterizations.

Finally, we compare the notions of p-strong equivalence and LL-strong equivalence. 
For \lpmln programs $P$ and $Q$, 
by Definition \ref{def:lpmln-strong-equivalence-p} and \ref{def:strong-equivalence-ll19}, 
it is easy to check that $P \equiv_{s, p} Q$ implies $P \equiv_{s, s} Q$, 
while $P \equiv_{s, pr} Q$ may not imply $P \equiv_{s, st} Q$, 
which is due to the complexity of definition of probability degree of an interpretation in \lpmlnend. 
By Equation \eqref{eq:probability-sm}, there are three kinds of interpretations for an \lpmln program $Q$:
\begin{itemize}
    \item the probabilistic stable models in $\psm{Q}$, i.e., the stable models of $Q$ that satisfy the most number of hard rules; 
    \item the non-probabilistic stable models in $\lsm{Q} - \psm{Q}$,  i.e., the stable models of $Q$ that do not satisfy the most number of hard rules; 
    \item the other interpretations. 
\end{itemize}
Among these interpretations, only the probabilistic stable models have the non-zero probability degrees w.r.t. an \lpmln program. 
According to the definition of LL-strong equivalence, we only know that two LL-strongly equivalent \lpmln programs have the same probabilistic stable models under any extensions, 
which means the programs may not be structurally equivalent. 
Therefore, by the definitions of p-strong equivalence and LL-strong equivalence, 
it shows that the p-strong equivalence implies the LL-strong equivalence, while the inverse may not hold. 
In \cite{Lee2019LPMLNSE}, Lee and Luo present a characterization for the LL-strong equivalence, 
which is shown in Theorem \ref{thm:p-se-ll19}.
\begin{thmC}[{\cite[Theorem 2]{Lee2019LPMLNSE}}]
\label{thm:p-se-ll19}
\lpmln programs $P$ and $Q$ are  LL-strongly equivalent iff they are  structurally equivalent, 
and there is a w-expression $c = exp(c_1 + c_2 * \alpha)$ such that 
$W(P, X) = c * W(Q, X)$ for any interpretation $X$.
\end{thmC}
Since $W(P, (X, Y)) = W(P, Y)$ for an SE-interpretation $(X, Y)$ and an \lpmln program $P$, 
by Theorem \ref{thm:lpmln-strong-equivalence-w} and \ref{thm:p-se-ll19}, 
it is obvious that the notions of p-strong equivalence and LL-strong equivalence are equivalent to each other. 
But in above discussion, it has been shown that the statement $P \equiv_{s, pr} Q$ implies $P \equiv_{s, st} Q$ is not a straightforward result. 
Unfortunately, in \cite{Lee2019LPMLNSE}, there is not sufficient discussion for the statement.

So far, we have defined the notion of p-strong equivalence and presented an SE-model approach to characterizing the notion. 
And we have investigated the relationships between the results in this section and the results in \cite{Lee2019LPMLNSE}. 
In the following sections, we will investigate properties of the p-strong equivalence from four aspects: 
(1) we present two relaxed notions of the p-strong equivalence and discuss their characterizations; 
(2) we analyze the computational complexities of deciding strong equivalences for \lpmln programs; 
(3) we investigate the relationships among the p-strong equivalence and the strong equivalences for ASP with weak constraints and ordered disjunctions; 
(4) we show the use of the p-strong equivalence in simplifying \lpmln programs. 

\section{Two Relaxed Notions of P-Strong Equivalence}
\label{sec:pse-varints}
The notion of p-strong equivalence requires that two \lpmln programs are p-ordinarily equivalent under any extension. 
But for program rewriting in many scenarios, the p-strong equivalence is somewhat strict. 
In this section, we present two relaxed notions of the p-strong equivalence and discuss their characterizations.

\subsection{P-Strong Equivalence under Soft Stable Model Semantics of \texorpdfstring{\lpmlnend}{LPMLN}}
By the \lpmln semantics, a stable model is allowed to violate any rules including hard rules. 
But for knowledge modeling of some problems, 
hard rules are usually used to encode definite knowledge, 
which means a solution to a problem should satisfy all hard rules. 
\begin{exa}
\label{ex:schedule-app}
Recall Example \ref{ex:inference-task-case}, \lpmln programs $P'$ and $Q'$ are obtained from $P$ and $Q$ in the example by replacing the constraints \eqref{con:1} and \eqref{con:2} with following hard constraint
\begin{equation}
\alpha ~:~ \leftarrow a, b. 
\end{equation}
Here, the atom $a$ represents ``play tennis'', and the atom $b$ represents ``play badminton'', 
and the programs $P'$ and $Q'$ can be viewed as a part of a scheduling application. 
Our next activity is one of ``play tennis'' and ``play badminton'', therefore, 
the hard rules of $P'$ and $Q'$ must be satisfied by any valid plan, 
which means $\{a, b\}$ cannot be a valid stable model under the case. 
\end{exa}

For the case, Lee and Wang present an extended semantics of \lpmln \cite{Lee2016Weighted}, called \textit{soft stable model semantics} (SSM semantics), which requires hard rules must be satisfied by stable models. 
In this section, we investigate the p-strong equivalence under the SSM semantics, called the sp-strong equivalence. 
Firstly, we review the SSM semantics of \lpmlnend.  
For an \lpmln program $P$, a \textit{soft stable model} $X$ of $P$ is a stable model of $P$ that satisfies all hard rules in $P$. 
By $\ssm{P}$, we denote the set of all soft stable models of $P$, 
i.e. $\ssm{P} = \{X \in \lsm{P} ~|~ X \models P^h \}$. 
For a soft stable model $X$, the weight degree $W_s(P, X)$ of $X$ w.r.t. $P$ is defined as 
\begin{equation}
\label{eq:soft-weight}
W_s(P, X) = W(P^s, X) = exp\left( \sum_{w:r \in P^s_X} w \right)
\end{equation}
and the probability degree $Pr_s(P, X)$ of $X$ w.r.t. $P$ is defined as 
\begin{equation}
\label{eq:soft-probability}
Pr_s(P, X) = \frac{W_s(P, X)}{\sum_{X' \in \ssm{P}} W_s(P, X')}
\end{equation}
Recall Example \ref{ex:lpmln-inference}, one can check that the empty set $\emptyset$ is not a soft stable model of the program $P$ in Example \ref{ex:lpmln-inference}, since $\emptyset$ does not satisfy the hard rule of $P$. 
It is worth noting that an \lpmln program $P$ always has stable models, while $P$ may not have soft stable models. 
For example, $\emptyset$ is a stable model of any \lpmln program, but if the hard rules of an \lpmln program are inconsistent, the program does not have soft stable model. 
In addition, the notion of soft stable model is different from the notion of probabilistic stable model in general.  
Since the former is required to satisfy all hard rules, 
while the latter is required to satisfy the most number of hard rules, 
which means if an \lpmln program $P$ has soft stable models, then $\psm{P} = \ssm{P}$ and for any soft stable model $X \in \ssm{P}$, $Pr(P, X) = Pr_s(P, X)$.

Secondly, we investigate the sp-strong equivalence by extending the SE-model approach for characterizing the p-strong equivalence under the original semantics. 
The sp-strong equivalence is defined as follows.

\begin{defi}
\label{def:p-se-soft-sm}
Two \lpmln programs $P$ and $Q$ are sp-strongly equivalent, denoted by $P \equiv_{s, sp} Q$, if for any \lpmln program $R$, $\ssm{P \cup R} = \ssm{Q \cup R}$, and for each soft stable model $X \in \ssm{P \cup R}$, we have $Pr_s(P \cup R, X) = Pr_s(Q \cup R, X)$.
\end{defi}

The notion of semi-strong equivalence can be extended to the SSM semantics naturally, 
i.e. $P$ and $Q$ are semi-strongly equivalent under the SSM semantics, denoted by $P \equiv_{s, ss} Q$, 
if for any \lpmln program $R$, $\ssm{P \cup R} = \ssm{Q \cup R}$. 
Accordingly, the notion of SE-model under the SSM semantics is defined as follows. 
\begin{defi}[soft SE-model]
\label{def:soft-se-model}
For an \lpmln program $P$, an SE-interpretation $(X, Y)$ is an SE-model of $P$ under the SSM semantics, called a soft SE-model, if $(X, Y) \in \lse{P}$ and $Y \models P^h$. 
\end{defi}
By $\sse{P}$, we denote the set of all soft SE-models of an \lpmln program $P$. 
The weight of a soft SE-model is defined as $W_s(P, (X, Y)) = W(P^s_Y)$. 
Based on the extended notions, Lemma \ref{lem:semi-pse-ssm} provides a characterization for the semi-strong equivalence under the SSM semantics. 

\begin{lem}
\label{lem:semi-pse-ssm}
Two \lpmln programs $P$ and $Q$ are semi-strongly equivalent under the SSM semantics iff $\sse{P} = \sse{Q}$. 
\end{lem}

\begin{proof}
For the if direction, we show that two \lpmln programs $P$ and $Q$ are semi-strongly equivalent under the SSM semantics if $\sse{P} = \sse{Q}$. 
We use proof by contradiction. 
Assume there is an \lpmln program $R$ and an interpretation $Y$ such that $Y \in \ssm{P \cup R}$ and $Y \not\in \ssm{Q \cup R}$, 
we have  $(Y, Y) \in \sse{P}$. 
Since $\sse{P} = \sse{Q}$ and $Y \not\in \ssm{Q \cup R}$, we have $(Y, Y) \in \sse{Q}$ and there is a proper subset $X$ of $Y$ such that $X \models \left( \overline{Q^h \cup R^h} \right)^Y \cup \lglred{(Q^s \cup R^s)}{Y}$, 
which means $(X, Y) \in \sse{Q}$. 
Similarly, we have $(X,Y) \in \sse{P}$,  
therefore, $X \models \left( \overline{P^h \cup R^h} \right)^Y \cup \lglred{(P^s \cup R^s)}{Y}$, which contradicts with the assumption.
Therefore, we have proven the if direction of Lemma \ref{lem:semi-pse-ssm}. 

For the only-if direction, we show that $\sse{P} = \sse{Q}$ if \lpmln programs $P \equiv_{s, ss} Q$. 
We use proof by contradiction. 
Assume there is an SE-interpretation $(X, Y)$ such that $(X,Y) \in \sse{P}$ and $(X, Y) \not \in \sse{Q}$. 
By the definition of soft SE-models, there are two cases: (1) $Y \not\models Q^h$; and (2) $X \not\models (\overline{Q^h})^Y \cup \lglred{Q^s}{Y}$. 

\textbf{Case 1.} 
If $Y \not\models Q^h$, let $R$ be an \lpmln program such that $R = \{\alpha : a. ~|~ a \in Y\}$. 
It is easy to check that $Y$ is a soft stable model of $P \cup R$, since $P$ and $Q$ are semi-strongly equivalent, $Y$ is also a soft stable model of $Q \cup R$, which contradicts with $Y \not\models Q^h$. 

\textbf{Case 2.}
If $Y \models Q^h$ and $X \not\models (\overline{Q^h})^Y \cup \lglred{Q^s}{Y}$, we can derive that $X$ is a proper subset of $Y$. 
Now we show that there is an \lpmln program $R$ such that $\ssm{P \cup R} \neq \ssm{Q \cup R}$ under the assumption. 
Let $R$ be an \lpmln program of the form \eqref{eq:sp-se-proof-program-r}. 
\begin{equation}
\label{eq:sp-se-proof-program-r}
R = \{\alpha : a. ~|~ a \in X\} \cup \{\alpha : a \leftarrow b. ~|~ a, b \in Y -X\}
\end{equation}
It is easy to check that $Y \models R$ and  $X \models R$, and for any other proper subset $X'$ of $Y$, $X' \not\models R$. 
Therefore, we can derive that for any proper subset $X'$ of $Y$, $X' \not\models (\overline{Q^h})^Y \cup \lglred{Q^s}{Y} \cup \overline{R}$, which means $Y$ is a soft stable model of $Q \cup R$. 
Since $(X, Y)$ is a soft SE-model of $P$, we have $X \models P^h \cup P^s_Y$. 
Since $X \models R$, we have $X \models P^h \cup P^s_Y \cup R$, which means $Y$ is not a stable model of $P \cup R$. 
It contradicts with the premise that $P \equiv_{s, ss} Q$.

Combining above results, Theorem \ref{thm:pse-soft-sm} is proven.
\end{proof}
Based on Lemma \ref{lem:semi-pse-ssm}, Theorem \ref{thm:pse-soft-sm} provides a characterization for the sp-strong equivalence. 
The proof of the theorem is similar to the proof of Theorem \ref{thm:lpmln-strong-equivalence-w}, 
therefore, we omit the details for brevity.
\begin{thm}
\label{thm:pse-soft-sm}
Two \lpmln programs $P$ and $Q$ are sp-strongly equivalent iff $\sse{P} = \sse{Q}$, 
and there exists a real number $c$ such that for any soft SE-model $(X, Y) \in \sse{P}$, $W_s(P, (X, Y)) = c * W_s(Q, (X, Y))$. 
\end{thm}

\begin{exa}
\label{ex:pse-soft-sm}
Recall \lpmln programs $P'$ and $Q'$ in Example \ref{ex:schedule-app}, 
we have known that they are neither p-strongly equivalent nor semi-strongly equivalent. 
But for the unweighted programs $\overline{P'} = \{a \vee b. ~ \leftarrow a, b. \}$ and $\overline{Q'} =\{a \leftarrow not ~ b. ~ b \leftarrow not ~a. ~ \leftarrow a, b. \}$, 
it is well-known that they are strongly equivalent under the ASP semantics \cite{Lifschitz2001Strongly}, 
therefore, $P'$ and $Q'$ should be sp-strongly equivalent. 
By Theorem \ref{thm:pse-soft-sm}, one can check that $P'$ and $Q'$ have the same soft SE-models and the weight distribution of soft SE-models coincides. 
Let $S_1 = (\{a\}, \{a\})$ and $S_2 = (\{b\}, \{b\})$, 
we have 
\begin{itemize}
    \item $\sse{P'} = \sse{Q'} = \{S_1, S_2\}$, 
    \item $W_s(P', S_1) = W_s(Q', S_1) = e^0 = 1$, and 
    \item $W_s(P', S_2) = W_s(Q', S_2) = e^0 = 1$. 
\end{itemize}
Therefore, under the SSM semantics, 
the programs $P'$ and $Q'$ are p-strongly equivalent.
\end{exa}

\begin{exa}
\label{ex:pse-soft-sm-2}
Consider following \lpmln programs $P''$
\begin{eqnarray}
2 &:& a \vee b. \\
\alpha &:& \leftarrow a, b. 
\end{eqnarray}
and $Q''$
\begin{eqnarray}
1 &:& a \leftarrow  not ~b. \\
1 &:& b \leftarrow not ~ a. \\
\alpha &:&~ \leftarrow a, b. 
\end{eqnarray}
It is easy to observe that $P''$ and $Q''$ are obtained by changing the weights of rules of programs $P$ and $Q$ in Example \ref{ex:se-model}. 
As we know, $P$ and $Q$ are not semi-strongly equivalent iff $\{\emptyset, \{a, b\}\}$ is not a common SE-model of $P$ and $Q$. 
It can be checked that all SE-interpretations of the form $(X, \{a, b\})$ are not soft SE-models of $P''$ and $Q''$, 
therefore, they are semi-strongly equivalent under the SSM semantics. 
For the soft SE-models, their weights are 
\begin{itemize}
    \item $W_s(P'', (\emptyset, \emptyset)) = W_s(Q'', (\emptyset, \emptyset)) = e^0$, 
    \item $W_s(P'', (\{a\}, \{a\})) = W_s(Q'', (\{a\}, \{a\})) = e^2$, and 
    \item $W_s(P'', (\{b\}, \{b\})) = W_s(Q'', (\{b\}, \{b\})) = e^2$. 
\end{itemize}
Therefore, the programs $P''$ and $Q''$ are sp-strongly equivalent. 
\end{exa}

Example \ref{ex:pse-soft-sm} and Example \ref{ex:pse-soft-sm-2} show that two \lpmln programs that are not p-strongly equivalent under the original \lpmln semantics could be p-strongly equivalent under the SSM semantics. 
In addition, from Definition \ref{def:p-se-soft-sm} and above examples, it can be observed that the notion of sp-strong equivalence can be viewed as a unified framework to investigate the strong equivalences in ASP and \lpmlnend. 
That is, if two \lpmln programs only contain hard rules, 
the sp-strong equivalence is reduced to the strong equivalence under the ASP semantics, 
if two \lpmln programs only contain soft rules, the sp-strong equivalence is reduced to the p-strong equivalence under the original \lpmln semantics. 
In addition, it is clear that some important results on the strong equivalence for ASP can be introduced to the sp-strong equivalence straightforwardly. 
For example, Lin and Chen have found a sufficient and necessary syntactic conditions for characterizing several classes of ASP programs \cite{Lin2005Discover}. 
Under the SSM semantics, these conditions can be directly used to decide the p-strong equivalence. 
Note that in the rest of the paper, unless we specifically point out, the p-strong or semi-strong equivalence means two \lpmln programs are p-strongly or semi-strongly equivalent under the original \lpmln semantics.

\subsection{Probabilistic Uniform Equivalence}
In recent years, knowledge graphs based applications are especially concerned. 
Generally, a knowledge graph is about the entities, their semantic types, 
properties, and relationships between entities \cite{Ehrlinger2016KG}. 
From the view of logic programming, a knowledge graph can be regraded as a set of facts, and these facts usually evolve over time. 
By introducing some rules, a knowledge graph becomes more powerful for modeling complex relations and inferring hidden knowledge, 
and these rules are not updated frequently by contrast. 
For example, in a knowledge graph about family members, 
the fact $child(joe, tom)$ represents ``joe is a child of tom'', 
and the blood relationship $blood$ can be defined as follows.  
\begin{eqnarray}
\alpha &:& blood(X, Y) \leftarrow blood(Y, X). \\
w &:& blood(X, Y) \leftarrow child(X, Y). \label{rule:child-blood}\\
\alpha &:& blood(X, Z) \leftarrow blood(X, Y), blood(Y, Z). 
\end{eqnarray}
where rule \eqref{rule:child-blood} is a soft uncertain rule, since a child may be adopted. 
The facts about family members will change constantly, but the definition of blood relationship is normally invariable.  
For program rewriting in the field of knowledge graph, the notion of p-strong equivalence is too strict, 
since we only concern the equivalence between programs extended by facts, 
which is called uniform equivalence in ASP \cite{Eiter2003UE}.  
In this section, we investigate the p-uniform equivalence for \lpmln programs. 
The notion of p-uniform equivalence is defined as follows. 

\begin{defi}[p-uniform equivalence]
\label{def:uniform-equivalence-p}
Two \lpmln programs $P$ and $Q$ are p-uniformly equivalent, denoted by $P \equiv_{u,p} Q$, if for any set $R$ of weighted facts, 
the programs $P \cup R$ and $Q \cup R$ are p-ordinarily equivalent. 
\end{defi}

The notion of semi-uniform equivalence can be defined naturally, i.e. $P$ and $Q$ are semi-uniformly equivalent, denoted by $P \equiv_{u,s} Q$, if for any  set $R$ of weighted facts, $\lsm{P \cup R} = \lsm{Q \cup R}$.
Now, we investigate the characterization of p-uniform equivalence between two \lpmln programs by introducing the notion of UE-models for \lpmln programs. 
Note that we only consider finite \lpmln programs in this paper. 

\begin{defi}[UE-models for \lpmlnend]
	\label{def:lue-model}
	For an \lpmln program $P$, a UE-model $(X,Y)$ of $P$ is an SE-model of $P$ satisfying
	\begin{itemize}
		\item[-] $X = Y$; or 
		\item[-] $X \subset Y$, and for any interpretation $X'$ satisfying $X \subset X' \subset Y$,  $(X',Y) \not\in \lse{P}$.
	\end{itemize}
\end{defi} 

By $\lue{P}$ we denote the set of all UE-models of an \lpmln program $P$. 
Actually, a UE-model of $P$ is also an SE-model of $P$, 
therefore, the weight degree of a UE-model is defined the same as that of SE-model. 
According to Definition \ref{def:lue-model}, we have following property of UE-models, which is used in the characterization of the p-uniform equivalence for \lpmlnend. 

\begin{lem} 
\label{lem:lue-sm}
For an \lpmln program $P$ and a non-total UE-model $(X, Y)$ of $P$, 
let $Z$ be an interpretation such that $X \subset Z \subset Y$, 
and  $R$ a set of weighted facts such that $\overline{R} = Z$, 
we have $Y \in \lsm{P \cup R}$.
\end{lem}

\begin{proof}
Since $(X,Y)$ is a non-total UE-model of an \lpmln program $P$, 
we have $X \models \lglred{P}{Y}$, 
and for any set $X'$ such that $X \subset X' \subset Y$, $X' \not\models \lglred{P}{Y}$. 
To check whether $Y$ is a stable model of the \lpmln program $P \cup R$, we use proof by contradiction. 
Assume that $Y$ is not a stable model of $P \cup R$, 
which means there is a proper subset $Y'$ of $Y$ such that $Y' \models \lglred{(P \cup R)}{Y}$. 
By the construction of $R$, it is obvious that $X \subset Z \subseteq Y'$. 
Since $(X, Y)$ is a UE-model of $P$, we have $Y' \not\models \lglred{P}{Y}$, 
which contradicts with $Y' \models \lglred{(P \cup R)}{Y}$. 
Therefore, $Y$ is a stable model of $P \cup R$.
\end{proof}

Following our approach to characterizing other notions of strong equivalences for \lpmlnend, 
we present a characterization of the semi-uniform equivalence between two \lpmln programs firstly.   
Then, we present a characterization of p-uniform equivalence on the basis of semi-uniform equivalence. 
Following lemma provides a characterization for semi-uniform equivalence between \lpmln programs. 

\begin{lem}
\label{lem:lpmln-sm-uniform-equiv-lue}
Two \lpmln programs $P$ and $Q$ are semi-uniformly equivalent, iff $\lue{P} = \lue{Q}$.
\end{lem}

\begin{proof}
By Lemma \ref{lem:XX-lse-model-LM} and the definition of UE-models, 
we only need to prove the case that $(X, Y)$ is a non-total UE-model, i.e. $X \subset Y$.

For the if direction, suppose $\lue{P} = \lue{Q}$, we need to prove that for any set $R$ of weighted facts, $\lsm{P \cup R} = \lsm{Q \cup R}$. 
We use proof by contradiction. 
For an interpretation $X$, without loss of generality, assume that $X \in \lsm{P \cup R}$ and $X \not\in \lsm{Q \cup R}$, 
we have $X \models \lglred{(P \cup R)}{X}$, $X \models \lglred{(Q \cup R)}{X}$, 
and there is a proper subset $X'$ of $X$ such that $X' \models \lglred{(Q \cup R)}{X}$. 
Hence, we have $(X',X) \in \lse{Q}$, the rest of the proof of if direction is divided into two cases.

\textbf{Case 1.}
If $(X',X) \in \lue{Q}$, we have $(X',X) \in \lue{P}$, hence, $X' \models \lglred{P}{X}$, which means $X$ cannot be a stable model of $P \cup R$.  
It contradicts with the assumption.

\textbf{Case 2.}
If $(X',X) \not\in \lue{Q}$, by the definition, there must be a subset $Z$ of $X$ such that $X' \subset Z$ and $(Z,X) \in \lue{Q}$. 
Since $\lue{P} = \lue{Q}$,  we have $(Z,X) \in \lue{P}$.
Since $X' \subset Z$, we have $Z \models \lglred{R}{X}$ and $Z \models \lglred{(P \cup R)}{X}$, which contradicts with $X \in \lsm{P \cup R}$.

Combining above results, the if direction of Lemma \ref{lem:lpmln-sm-uniform-equiv-lue} is proven.

For the only-if direction, suppose for any set $R$ of weighted facts, $\lsm{P \cup R} = \lsm{Q \cup R}$, 
we need to prove that $\lue{P} = \lue{Q}$.
We use proof by contradiction.
Without loss of generality, assume that there exists a non-total SE-model $(X,Y)$ of $P$ such that $(X,Y) \in \lue{P}$ and $(X, Y) \not\in \lue{Q}$. 
The rest of the proof of the only-if direction is divided into three cases. 

\textbf{Case 1.}
If for any interpretation $Z \subset Y$, $(Z,Y) \not\in \lse{Q}$, by Lemma \ref{lem:x-equilibrium-sm}, we have $Y \in \lsm{Q}$. 
Hence, we have $Y \in \lsm{P}$, which means there does not exist a proper subset $X'$ of $Y$ such that $X' \models \lglred{P}{Y}$. 
It contradicts with $(X,Y) \in \lue{P}$.

\textbf{Case 2.}
If $(X, Y) \in \lse{Q}$, 
by the definition of UE-model, there is an interpretation $Z$ such that $X \subset Z \subset Y$ and $(Z, Y) \in \lue{Q}$. 
Let $R'$ be an \lpmln program such that $\overline{R'} = Z$, by Lemma \ref{lem:lue-sm}, we have $Y \in \lsm{P \cup R'}$, 
therefore, $Y \in \lsm{Q \cup R'}$. 
Since $(Z,Y) \in \lue{Q}$, we have $Z \models \lglred{Q}{Y}$ and $Z \models \lglred{(Q \cup R')}{Y}$, which contradicts with $Y \in \lsm{Q \cup R'}$.

\textbf{Case 3.}
If $(X, Y) \not\in \lse{Q}$, 
we show there is an interpretation $Z$ such that $X \subset Z \subset Y$, $(Z, Y) \in \lse{Q}$.
We use proof by contradiction. 
Assume for any interpretation $Z$ such that $X \subset Z \subset Y$, $(Z, Y) \not\in \lse{Q}$, 
for any interpretation $X'$ such that $X \subseteq X' \subset Y$, we have $X' \not\models \lglred{Q}{Y}$. 
Let $R''$ be an \lpmln program such that $\overline{R''} = X$.
It is easy to check that for any proper subset $X''$  of $X$, $X'' \not\models \lglred{(Q \cup R'')}{Y}$. 
Combining the above results, we have $Y \in \lsm{Q \cup R''}$, which means $Y \in \lsm{P \cup R''}$. 
Since $(X, Y) \in \lue{P}$, we have $X \models \lglred{(P \cup R'')}{Y}$, 
which contradicts with  $Y \in \lsm{P \cup R''}$. 
Therefore, there is an interpretation $Z$ such that $X \subset Z \subset Y$, $(Z, Y) \in \lse{Q}$, 
and the rest of the proof of this case is the same as the proof of Case 2. 

Combining the above results, the only-if direction of Lemma \ref{lem:lpmln-sm-uniform-equiv-lue} is proven. 
\end{proof}

\begin{exa}
\label{ex:semi-ue}
Recall Example \ref{ex:se-model}, it is easy to check that the SE-interpretation $(\emptyset, U)$ is not a UE-model of the programs $P$ and $Q$, 
therefore, $P$ and $Q$ have the same UE-models.  
By Lemma \ref{lem:lpmln-sm-uniform-equiv-lue}, the programs $P$ and $Q$  are semi-uniformly equivalent.  
\end{exa}

Now, we investigate the characterization of p-uniform equivalence for \lpmln programs. 
Firstly, we present a sufficient condition for the characterization in Lemma \ref{lem:finite-lpmln-uniform-equivalence-lue}, then, we check whether the condition is necessary. 

\begin{lem}
\label{lem:finite-lpmln-uniform-equivalence-lue}
Two \lpmln programs $P$ and $Q$ are p-uniformly equivalent, if  $\lue{P} = \lue{Q}$, 
and there exist two constants $c$ and $k$ such that 
for each UE-model $(X,Y) \in \lue{P}$, $W(P,(X,Y)) = exp(c + k*\alpha) * W(Q,(X,Y))$.
\end{lem} 

The proof of Lemma \ref{lem:finite-lpmln-uniform-equivalence-lue} is similar to the proof of Lemma \ref{lem:lpmln-strong-equivalence-suf}, therefore, we omit the detail for brevity. 
Lemma \ref{lem:finite-lpmln-uniform-equivalence-lue} shows a sufficient condition to characterize the p-uniform equivalence between \lpmln programs, called PUE-condition. 
To show whether the PUE-condition is necessary, we also need to find a necessary extension just as we did in proving Theorem \ref{thm:lpmln-strong-equivalence-w}.
That is, for an \lpmln program $P$ and a set $U$ of literals such that $\olit{P}\subseteq U$, we need to find a set $E$ of \lpmln programs satisfying for any interpretations $X$ and $Y$ in $2^{U^+}$, 
there exists a program $R$ of $E$ such that both $X$ and $Y$ are probabilistic stable models of $P \cup R$, 
where the programs in $E$ are sets of weighted facts. 
Unfortunately, there does not exist such an extension for p-uniformly equivalent \lpmln programs, which is discussed via Proposition \ref{prop:property-ii-not-hold} and \ref{prop:pue-case2}. 

\begin{prop}
\label{prop:property-ii-not-hold}
Let $P$ be an arbitrary \lpmln program, 
for any two total UE-models $(X,X)$ and $(Y,Y)$ of $P$ such that $X \subset Y$ and $h(P, X) < h(P, Y)$, 
there does not exist a set $R$ of weighted facts such that $h(P \cup R, X) \geq h(P \cup R, Y)$. 
\end{prop}

\begin{proof}
For a set $R$ of weighted facts, $h(P \cup R, X) = h(P, X) + h(R, X) = |P^h_X| + |R^h_X|$ and $h(P \cup R, Y) = |P^h_Y| + |R^h_Y|$. 
By the definition of \lpmln reduct, we have 
\begin{equation}
    R^h_X = \{\alpha : r \in R ~|~ h(r) \cap X \neq \emptyset \} \text{ and } 
    R^h_Y = \{\alpha : r \in R ~|~ h(r) \cap Y \neq \emptyset \}
\end{equation}
Since $X \subset Y$, we have $R^h_X \subseteq R^h_Y$, which means $h(P \cup R, X) < h(P \cup R, Y)$. 
Therefore,  Proposition \ref{prop:property-ii-not-hold} is proven.
\end{proof}
Proposition \ref{prop:property-ii-not-hold} shows that necessary extensions of weighted facts may not exist for some \lpmln programs. 
Recall the \lpmln program $P$ in Example \ref{ex:se-model}, 
from Table \ref{tab:se-models}, 
we can observe that $h(P, \emptyset) < h(P, \{a\})$, 
and for any set $R$ of weighted facts, $h(P \cup R,  \emptyset) = 0$, while $h(P \cup R, \{a\}) \geq 1$.
That is, $h(P \cup R, \emptyset)$ is always less than $h(P \cup R, \{a\})$, 
which means the condition on the weight of $\emptyset$ is not necessary for characterizing the p-uniform equivalence of $P$. 

In addition, there are other complicated situations that prevent a non-probabilistic stable model from becoming probabilistic stable model of any extended program. 
For example, suppose interpretations $X$ and $Y$ are stable models of an \lpmln program $P$ such that $X \not\in \psm{P}$ and $Y \in \psm{P}$. 
To adjust the weight degrees of $X$ and $Y$, we can add a weighted fact $\alpha:r$ to $P$, and the fact should satisfy following conditions:
\begin{itemize}
	\item $\alpha : r \not\in P$, which means there are newly introduced atoms in $r$ in general, i.e. $lit(r) \not\subseteq \olit{P}$; 
	\item $h(r) \cap X \neq \emptyset$ and $h(r) \cap Y = \emptyset$; and 
	\item combining above two conditions, $r$ should be a disjunctive fact.
\end{itemize}
Suppose fact ``$\alpha : a \vee c.$'' satisfies all of above conditions, and $a$ belongs to $X$. 
By $P'$, we denote the extended program $P \cup \{\alpha : a \vee c.\}$.
It is easy to check that both $X$ and $Y$ are stable models of $P'$, and $h(P', X) = h(P, X) + 1$ and $h(P', Y) = h(P, Y)$. 
It seems $X$ could become a probabilistic stable model of an extended program by adding such kinds of weighted facts. 
However, there are some new stable models of $P'$ such as $Y' = Y \cup \{c\}$, 
it is easy to check that $h(P', Y') = h(P, Y) + 1$, 
which means the difference between the numbers of hard rules satisfied by $X$ and a probabilistic stable model cannot be decreased. 
Obviously, once we introduce new atoms to an extension, $X$ cannot become a probabilistic stable model of the extended program. 
A general case of above discussion is shown in Proposition \ref{prop:pue-case2}.

\begin{prop}
\label{prop:pue-case2}
For an \lpmln program $Q$ and two stable models $X$ and $Y$ of $Q$ such that $X \not\in \psm{Q}$ and $Y \in \psm{Q}$,  
if $R$ is a set of weighted facts such that for each rule $\alpha : r \in R$, $lit(r) \not\subseteq lit(\overline{Q})$, 
then at least one of $X$ and $Y$ cannot be a probabilistic stable model of $Q \cup R$. 
\end{prop}

\begin{proof}
Since soft rules do not contribute to whether an interpretation is a probabilistic stable model, we assume $R$ only contains hard rules. 
For the probabilistic stable model $Y$, there are two cases. 
(1) If $Y \models R$, then $h(Q \cup R, Y) = h(Q, Y) + |R| > h(Q, X) + |R_X|$, i.e. $h(Q \cup R, Y) > h(Q \cup R, X)$. 
Therefore, $X$ cannot be a probabilistic stable model of $Q \cup R$. 
(2) If $Y \not\models R$, since $\olit{R} \not\subseteq \olit{Q}$, there must be a minimal interpretation $Z$ such that $Y \subset Z$, $(Z - Y) \subseteq \olit{R - R_Y}$, and $Z \models R$. 
It is easy to check that $Z$ is a stable model of $Q \cup R$ and $R_Y \subset R_Z$, 
therefore, $h(Q \cup R, Z) > h(Q \cup R, Y)$ and $Y$ is not a probabilistic stable model of $Q \cup R$. 
Combining above results, Proposition \ref{prop:pue-case2} is proven. 
Besides, by Proposition \ref{prop:property-ii-not-hold}, if  $Y \not\models R$, there does not exist a set $R'$ of weighted facts such that $Y \in \psm{Q \cup R \cup R'}$. 
\end{proof}
In a word, for the p-uniform equivalence, we only find a sufficient condition to characterize it, i.e. the PUE-condition, 
and there are many complicated cases such that the PUE-condition is not necessary. 
For arbitrary \lpmln programs, there may not exist a general necessary condition to characterize the p-uniform equivalence. 
But for \lpmln programs containing no hard rules, it is obvious that the PUE-condition is sufficient and necessary, which is shown as follows. 

\begin{thm}
\label{thm:pue-soft-rule}
For \lpmln programs $P$ and $Q$ that do not contain hard rules, $P$ and $Q$ are p-uniformly equivalent, iff $\lue{P} = \lue{Q}$, and there exists a real number $c$ such that for each UE-model $(X, Y) \in \lue{P}$, $W(P, (X, Y)) = c* W(Q, (X,Y))$.
\end{thm}

\begin{exa}
\label{ex:p-ue}
Continue Example \ref{ex:semi-ue}, it is easy to check that programs $P$ and $Q$ have the same UE-models, and for each UE-model $(X, Y)$ such that $Y \neq \emptyset$, $W(Q, (X, Y)) = e^\alpha * W(P, (X, Y))$. 
Although $P$ and $Q$ do not satisfy the PUE-condition of Lemma \ref{lem:finite-lpmln-uniform-equivalence-lue}, $P$ and $Q$ are still p-uniformly equivalent. 
Since $\emptyset$ cannot be a probabilistic stable model of any extended programs of $P$ and $Q$ by Proposition $\ref{prop:property-ii-not-hold}$, the condition of weight on UE-model $(\emptyset, \emptyset)$ is not necessary. 
\end{exa}

\subsection{Relationships}

\begin{figure}
	\centering
	\includegraphics{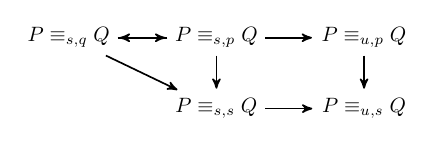}
	\caption{Relationships among Strong Equivalences for \lpmlnend}
	\label{fig:equiv-relations}
\end{figure}

The relationships among different notions of strong equivalences have been discussed separately in previous sections, which is summarized as follows. 
As shown in Figure \ref{fig:equiv-relations}, for \lpmln programs $P$ and $Q$,  
\begin{itemize}
	\item $P \equiv_{s, p} Q$ implies  $P \equiv_{u, p} Q$, and the inverse does not hold in general; 
	\item $P \equiv_{s, p} Q$ iff $P \equiv_{s, q} Q$;
	\item for each $\Delta \in \{p, q\}$, $P \equiv_{s, \Delta} Q$ implies $P \equiv_{s, s} Q$, and the inverses do not hold in general; 
	\item both $P \equiv_{s,s} Q$ and $P \equiv_{u, p} Q$ imply $P \equiv_{u, s} Q$, and the inverses do not hold in general.
\end{itemize}
For \lpmln programs containing only hard rules, they are sp-strongly equivalent, 
iff their unweighted ASP counterparts are strongly equivalent under the ASP semantics; 
and for \lpmln programs containing only soft rules, they are sp-strongly equivalent, iff they are p-strongly equivalent. 

In addition, due to the similarity between the definitions of SE-models for \lpmln and ASP, we consider the relationship between the semi-strong equivalence for \lpmln and the strong equivalence for ASP. 
For an \lpmln program $P$, it is easy to check that an SE-model $(X, Y) \in \lse{P}$ is also an SE-model of ASP program $\overline{P_Y}$, 
while an SE-model of $\overline{P_Y}$ is not an SE-model of $P$ in general. 
Since for an SE-model $(X', Y')$ of $\overline{P_Y}$ such that $Y' \neq Y$, 
$P_Y$ is not equal to $P_{Y'}$ usually, which means $(X', Y')$ may not be an SE-model of $P$. 
Therefore, \lpmln programs $P$ and $Q$ are semi-strongly equivalent does not mean their \lpmln reducts w.r.t. an interpretation are strongly equivalent under the ASP semantics. 
Recall Example \ref{ex:p-strong-equivalence}, programs $P = \{ w_1 : a \vee b. ~ w_2 : b \leftarrow a.\}$ and $Q = \{w_3 : b. ~ w_4 : a \leftarrow  not ~b. \}$ are semi-strongly equivalent. 
For an interpretation $I = \{a\}$, $\overline{P_I} = \{ a \vee b. \}$ and $\overline{Q_I} = \{ a \leftarrow not ~ b.  \}$. 
It is easy to check that $\overline{P_I}$ has two stable models, while $\overline{Q_I}$ only has one, therefore, they are not strongly equivalent.

\section{Computational Complexities}
In this section,  we discuss the computational complexities of deciding strong equivalences for \lpmlnend. 
As a by-product, we present a brief discussion on the potential application of the method presented in this section. 

\subsection{Computational Complexity Results}
Since checking SE-model and checking weight degrees are two separate processes, 
we consider the semi-strong equivalence checking firstly, 
which is shown as follows. 

\begin{thm}
\label{thm:complexity-se}
For \lpmln programs $P$ and $Q$, deciding $P \equiv_{s, s} Q$ is co-NP-complete.
\end{thm}

\begin{proof}
\textbf{Hardness.}
To show co-NP-hardness, we provide a polynomial reduction of checking tautology to deciding the strong equivalence of two \lpmln programs. 
For a propositional formula in conjunctive normal form (CNF) $F$, it is well-known that checking whether $F$ is a tautology is co-NP-complete. 
In this paper, a CNF $F$ is a formula of the form \eqref{eq:cnf}, which is a conjunction of clauses. 
\begin{equation}
\label{eq:cnf}
F = \bigwedge_{i=1}^n C_i \text{, where } C_i = \{c_{i,1}, \ldots, c_{i, m_i}, \neg c_{i, m_i + 1}, \ldots, \neg c_{i, n_i}\} ~ (1 \leq i \leq n)
\end{equation}
For a CNF $F$ of the form \eqref{eq:cnf}, let $a$ and $b$ be newly introduced atoms.  
For each clause $C_i$ of $F$, 
$\psi_1(C_i)$ is an \lpmln rule of the form 
\begin{equation}
\alpha : a \leftarrow \neg c_{i, 1}, \ldots, \neg c_{i, m_i}, c_{i, m_i + 1}, \ldots, c_{i, n_i} ~not~ b.
\end{equation}
$\psi_2(C_i)$ is an \lpmln rule of the form 
\begin{equation}
\alpha : b \leftarrow \neg c_{i, 1}, \ldots, \neg c_{i, m_i}, c_{i, m_i + 1}, \ldots, c_{i, n_i} ~not~ a.
\end{equation}
and $\psi(F)$ is an \lpmln program as follows
\begin{equation}
\begin{split} \psi(F) =
& \{\alpha : a  \leftarrow not ~c, ~not~ \neg c. ~|~ c \in at(F)\} ~\cup \\
& \{\alpha : b  \leftarrow not ~c, ~not~ \neg c. ~|~ c \in at(F)\} ~\cup \\
& \{\alpha : c_1 \vee \neg c_1 \leftarrow not ~c_2, ~not~ \neg c_2. ~|~ c_1, ~c_2 \in at(F)\}
\end{split}
\end{equation}
where $at(F)$ is the set of atoms occurred in  $F$. 
Based on the above notations, we define \lpmln programs $\psi_1(F)$ and $\psi_2(F)$ as follows
\begin{eqnarray}
& \psi_1(F) = \{ \psi_1(C_i) ~|~ 1 \leq i \leq n \} \cup \psi(F) \\ 
& \psi_2(F) = \{ \psi_2(C_i) ~|~ 1 \leq i \leq n \} \cup \psi(F) 
\end{eqnarray}

Next, we show that a CNF $F$ is a tautology iff \lpmln programs $\psi_1(F)$ and $\psi_2(F)$ are semi-strongly equivalent. 
Firstly, we introduce some notions. 
Let $F$ be a CNF and $I$ an interpretation, if for any atom $a \in at(F)$, either $a \in I$ or $\neg a \in I$, we say $I$ is a total interpretation, 
otherwise, $I$ is a partial interpretation. 
For a CNF $F$ and a partial interpretation $I$, it is easy to check that the \lpmln reduct 
$\psi(F)_I$ and the GL-reduct $\lglred{\psi(F)}{I}$ are shown as follows.
\begin{equation}
\begin{split}
\psi(F)_I = &  \{\alpha : a  \leftarrow not ~c, ~not~ \neg c. ~|~ \{a, c, \neg c\} \cap I \neq \emptyset \} \\
& \cup \{\alpha : b  \leftarrow not ~c, ~not~ \neg c. ~|~ \{b, c, \neg c\} \cap I \neq \emptyset \} \\
& \cup \{\alpha : c_1 \vee \neg c_1 \leftarrow not ~ c_2, ~not~ \neg c_2. ~|~ \{c_1, \neg c_1, c_2, \neg c_2\} \cap I \neq \emptyset \}
\end{split}
\end{equation}
\begin{equation}
\lglred{\psi(F)}{I} =  \{a. ~|~ a \in I \} \cup \{b. ~|~ b \in I \} \cup 
\{c \vee \neg c. ~|~ \{c, \neg c\} \cap I \neq \emptyset \}
\end{equation}
Obviously, if $I$ is a partial interpretation, $(I, I)$ is the only SE-model of $\psi_1(F)$ and $\psi_2(F)$. 
Therefore, for the proof, we only need to consider the case that $I$ is a total interpretation. 
For a total interpretation $I$, it is easy to check that $\lglred{\psi(F)}{I} = \emptyset$, therefore, we only need to consider the rules of the form $\psi_1(C_i)$ and $\psi_2(C_i)$.

For the if direction, if $\psi_1(F) \equiv_{s, s} \psi_2(F)$, 
we have for any SE-interpretation $(X, Y)$, $X \models \lglred{\psi_1(F)}{Y}$ iff $X \models \lglred{\psi_2(F)}{Y}$, 
and we need to show $F$ is a tautology. 
We use proof by contradiction. 
Assume $F$ is not a tautology, there must be a total interpretation $X$ and a clause $C_k$ of $F$ such that $\{a, b\} \cap  X = \emptyset$ and $X \not\models C_k$. 
Let $Y = X \cup \{b\}$, 
it is easy to check that $\lglred{\psi_1(F)}{Y} = \emptyset$ and $\lglred{\psi_2(F)}{Y}  = \{\psi_2^+(C_i) ~|~ 1 \leq i \leq n \}$, 
where $\psi_2^+(C_i)$ is obtained from $\psi_2(C_i)$ by removing default literal ``$not ~ a$'' and weight ``$\alpha$''. 
Since $X \not\models C_k$ and $b \not\in X$, we have $X \models b(\psi_2^+(C_k))$ but $X \not\models h(\psi_2^+(C_k))$, i.e. $X \not\models \psi_2^+(C_k)$.
Therefore, we can derive that  $X \not\models \lglred{\psi_2(F)}{Y}$, while, 
it is obvious that $X \models  \lglred{\psi_1(F)}{Y}$, 
which contradicts with $\psi_1(F) \equiv_{s, s} \psi_2(F)$. 
Therefore, $F$ is a tautology.

For the only-if direction, if $F$ is a tautology, we have for any total interpretation $X$, 
$X \models C_i$ $(1 \leq i \leq n)$, which means $X$ does not satisfy the positive bodies of $\psi_1(C_i)$ and $\psi_2(C_i)$. 
For a total interpretation $X$, there are four cases: (1) $\{a, b\} \cap  X = \emptyset$; (2) $\{a, b\} \subset X$; (3) $a \in X$ and $b \not\in X$; and (4) $a \not\in X$ and $b \in X$. 

\textbf{Case 1.} 
If $\{a, b\} \cap X = \emptyset$, we have 
\begin{equation}
    \lglred{\psi_1(F)}{X} = \{ \psi^+_1(C_i) ~|~ 1 \leq i \leq n \}  \text{ and } 
    \lglred{\psi_2(F)}{X} = \{ \psi^+_2(C_i) ~|~ 1 \leq i \leq n \}
\end{equation}
Since $X$ does not satisfies the positive bodies of $\psi_1(C_i)$ and $\psi_2(C_i)$ for any $1 \leq i \leq n$, 
any subset $X'$ of $X$ does not satisfy the positive bodies of $\psi_1(C_i)$ and $\psi_2(C_i)$ for any $1 \leq i \leq n$ either, 
which means $X' \models \lglred{\psi_1(F)}{X}$ and $X' \models  \lglred{\psi_2(F)}{X}$. 
Therefore, for a total interpretation $X$ and any subset $X'$ of $X$, $(X', X)$ is an SE-model of $\psi_1(F)$ and $\psi_2(F)$.

\textbf{Case 2.} 
If $\{a, b\} \subset X$, we have 
\begin{equation}
    \lglred{\psi_1(F)}{X} = \emptyset  \text{ and } 
    \lglred{\psi_2(F)}{X} = \emptyset
\end{equation}
therefore, for any subset $X'$ of $X$, $(X', X)$ is an SE-model of $\psi_1(F)$ and $\psi_2(F)$. 

\textbf{Case 3.}
If $a \in X$ and $b \not\in X$, we have 
\begin{equation}
    \lglred{\psi_1(F)}{X} = \{ \psi^+_1(C_i) ~|~ 1 \leq i \leq n \}  \text{ and } 
    \lglred{\psi_2(F)}{X} = \emptyset
\end{equation}
From above discussion, it is obvious that for a total interpretation $X$ and any subset $X'$ of $X$, $(X', X)$ is an SE-model of $\psi_1(F)$ and $\psi_2(F)$.

\textbf{Case 4.}
$a \not\in X$ and $b \in X$, we have 
\begin{equation}
    \lglred{\psi_1(F)}{X} = \emptyset  \text{ and } 
    \lglred{\psi_2(F)}{X} = \{ \psi^+_2(C_i) ~|~ 1 \leq i \leq n \}
\end{equation}
From the above discussion, it is obvious that for a total interpretation $X$ and any subset $X'$ of $X$, $(X', X)$ is an SE-model of $\psi_1(F)$ and $\psi_2(F)$.

The above results prove that deciding the semi-strong equivalence for \lpmln programs is co-NP-hard.  

\textbf{Membership.}
To show the co-NP-membership, we provide a polynomial reduction of checking the semi-strong equivalence to the problem of checking tautology. 
For \lpmln programs $P$ and $Q$, $U$ is a universe of literals such that $\olit{P \cup Q} \subseteq U$. 
For a literal $u \in U$, by $\hat{u}$, we denote an atom w.r.t. $u$, i.e. for an atom $a$, $\hat{a} = a$ and $\hat{\neg a} = a'$, where $a'$ is a newly introduced atom.
By $\hat{U}$, we denote the set of atoms obtained from $U$, i.e. $\hat{U} = \{\hat{u} ~|~ u \in U\}$. 
By $a^*$, we denote a newly introduced atom w.r.t. an atom $a \in \hat{U}$.
For an ASP rule $r$ of the form \eqref{eq:asp-rule-form-abbr}, 
$\delta_1(r)$ is a propositional formula of the form 
\begin{equation}
 \bigwedge_{b \in b^+(r)} \hat{b}  \wedge  \bigwedge_{c \in b^-(r)} \neg \hat{c}^*  \rightarrow  \bigvee_{a \in h^+(r)} \hat{a} 
\end{equation}
and $\delta_2(r)$ is a propositional formula of the form 
\begin{equation}
\bigwedge_{b \in b^+(r)} \hat{b}^*  \wedge  \bigwedge_{c \in b^-(r)} \neg \hat{c}^*  \rightarrow  \bigvee_{a \in h^+(r)} \hat{a}^* 
\end{equation}
where ``$\rightarrow$", ``$\wedge$", and ``$\vee$" are logical entailment, conjunction, and disjunction in propositional logic respectively. 
And for an \lpmln program $P$, let $\Delta(P)$ be the propositional formula of the form (\ref{eq:lpmln-prop-trans}):
\begin{equation}
\label{eq:lpmln-prop-trans}
\bigwedge_{r \in \overline{P}} \left( \deltar{r} \right)
\end{equation}
And for a set $U$ of literals, let $\Gamma(U)$ be the propositional formula of the form (\ref{eq:gamma-a}):
\begin{equation}
\label{eq:gamma-a}
\bigwedge_{b \in \hat{U}} \left(b \rightarrow b^*\right) \wedge \bigwedge_{a \in U} \left( (\hat{a} \wedge \hat{\neg a} \rightarrow \bot ) \wedge (\hat{a}^* \wedge \hat{\neg a}^* \rightarrow \bot ) \right)
\end{equation}
where $\bot$ denotes ``false''. 

Next, we show that an SE-interpretation $(X, Y)$ is an SE-model of an \lpmln program $P$ iff $\phi((X, Y))$ is a model of $\ptrans{U}{P}$, where $\phi((X,Y))$ is a set of atoms constructed as follows
\begin{equation}
\label{eq:pmodel-to-lse}
\phi((X,Y)) = \hat{X} \cup \{\hat{y}^* ~|~ y \in Y\}
\end{equation}

For the if direction, suppose $Z$ is a model of $\ptrans{U}{P}$, $(X,Y)$ is constructed from $Z$ by the inverse of the map $\phi$:
\begin{eqnarray}
& X = \{a ~|~ \hat{a} \in Z\} \\
& Y = \{a ~|~ \hat{a}^* \in Z\}
\end{eqnarray}
Obviously, $X$ and $Y$ are consistent. 
We need to show that $(X,Y) \in \lse{P}$, which means $X \subseteq Y$ and $X \models \lglred{P}{Y}$. 
We use proof by contradiction. 

Assume $X \not\subseteq Y$, i.e. there is an atom $b \in \hat{X}$ such that $b^* \not\in Z$. 
It is easy to check that the formula $b \rightarrow b^*$ in formula \eqref{eq:gamma-a} cannot be satisfied by $Z$, and $Z$ is not a model of $\ptrans{U}{P}$, which contradicts the premise. Therefore, we have shown that $X \subseteq Y$. 

Assume $X \not\models \lglred{P}{Y}$, which means there is a rule $r \in \lglred{P}{Y}$ such that $b^+(r) \subseteq X$ and $h(r) \cap X = \emptyset$. 
Suppose rule $r$ is obtained from rule $r'$ by removing its negative body, 
by the definitions of \lpmln reduct and GL-reduct,  we have $b^-(r') \cap Y = \emptyset$ and $Y \models r'$. 
Since $X \subseteq Y$, we have $b^+(r) \subseteq Y$ and $h(r) \cap Y \neq \emptyset$. 
By the construction of $X$ and $Y$, it is easy to check that $Z \models \delta_2(r')$ and $Z \not\models \delta_1(r')$, which means $Z$ cannot be a model of $\ptrans{U}{P}$. Therefore, we have shown that $X \models \lglred{P}{Y}$. 

Combining the above results, we have shown that $(X, Y)$ is an SE-model of $P$.

For the only-if direction, suppose $(X,Y)$ is an SE-model of $P$, $Z = \phi((X,Y))$, we need to show that $Z$ is a model of $\ptrans{U}{P}$. 
Since $X \subseteq Y$ and both $X$ and $Y$ are consistent, by the construction of $Z$, it is easy to show that $Z \models \Gamma(U)$. 
For each rule $w:r \in P$, if $w:r \not\in P_Y$ i.e. $Y \not\models w:r$, we have $b^+(r) \subseteq Y$, $b^-(r) \cap Y = \emptyset$, and $h(r) \cap Y = \emptyset$. 
By the construction of $Z$, we have $Z \not\models \delta_2(r)$, hence, $Z \models \deltar{r}$. 
If rule $w:r \in P_Y$ i.e. $Y \models w:r$, we have $Z \models \delta_2(r)$. 
If $b^-(r) \cap Y \neq \emptyset$, it is easy to check that $Z \models \delta_1(r)$. Hence, $Z \models \deltar{r}$. 
If $b^-(r) \cap Y = \emptyset$, let $r'$ be the rule that is obtained from $r$ by removing its negative body. 
Since $X \models \lglred{P}{Y}$, we have $X \models r'$, which means $Z \models \delta_1(r)$. Therefore, $Z \models\deltar{r}$. 

Combining the above results, we have shown that $Z$ is a model of $\ptrans{U}{P}$.

Above results show that the semi-strong equivalence checking for \lpmln programs $P$ and $Q$ can be reduced to checking whether the propositional formula of the form \eqref{eq:lpmln-to-tautology} is a tautology, which is in co-NP.
\begin{equation}
    \label{eq:lpmln-to-tautology}
    (\ptrans{U}{P}) \leftrightarrow (\ptrans{U}{Q})
\end{equation}
Therefore, it proves the co-NP-membership of the semi-strong equivalence checking in \lpmlnend. 
\end{proof}
For \lpmln programs $P$ and $Q$, 
Theorem \ref{thm:complexity-se} shows that deciding $P \equiv_{s, s} Q$ is co-NP-complete. 
To decide $P \equiv_{s, p} Q$, we need to additionally check the relationships among the weights of SE-models. 
For an interpretation $X$, 
we can proceed rule by rule, and check whether each rule can be satisfied by $X$, 
therefore, computing $P_X$ and $Q_X$ is feasible in polynomial time, 
which means computing $W(P, X)$ and $W(Q, X)$ can also be done in polynomial time. 
By Lemma \ref{lem:XX-lse-model-LM}, every total SE-interpretation is an SE-model of an \lpmln program. 
Therefore, for the weights checking, we can guess two interpretations $X$ and $Y$ and check whether $W(P, X) / W(Q, X)  \neq W(P, Y) / W(Q, Y)$, 
which means deciding $P \equiv_{s, p} Q$ is in co-NP.
Since checking weights and checking semi-strong equivalence are independent, we have following results. 

\begin{cor}
\label{cor:complexity-se-pq}
For \lpmln programs $P$ and $Q$, deciding $P \equiv_{s, p} Q$ is co-NP-complete.
\end{cor}

For the p-strong equivalence under the SSM semantics, the weight checking is similar to above discussion, therefore, it is in co-NP. 
For the semi-strong equivalence under the SSM semantics, 
the proof of co-NP-hardnees of Theorem \ref{thm:complexity-se} is still available, therefore, deciding semi-strong equivalence under the SSM semantics is co-NP-hard. 
To show the co-NP-membership, we slightly modify the proof of co-NP-membership of Theorem \ref{thm:complexity-se}. 
The only difference between the two kinds of \lpmln semantics is that hard rules cannot be violated under the SSM semantics. 
Therefore, for an \lpmln program $P$, we define $\Delta'(P)$ as 
\begin{equation}
\Delta'(P) = \bigwedge_{r \in \overline{P^s}} ( \delta_2(r) \rightarrow \delta_1(r) ) \wedge  \bigwedge_{r \in \overline{P^h}} ( \delta_2(r) \wedge \delta_1(r) )
\end{equation}
which is a combination of our translation of handling weighed rules and Lin's translation of handling ASP rules \cite{Lin2002Reducing}. 
Similarly, there is a one-to-one mapping between the soft SE-models of $P$ and the models of $\Gamma(U) \wedge \Delta'(P)$, 
which shows the co-NP-membership of deciding semi-strong equivalence under the SSM semantics. 

\begin{cor}
	\label{cor:complexity-pse-ssm}
	For \lpmln programs $P$ and $Q$, both of deciding $P \equiv_{s, s} Q$ and $P \equiv_{s, p} Q$ under the SSM semantics  are co-NP-complete.
\end{cor}

For the uniform equivalence checking, it is obviously harder than strong equivalence checking, 
which can be seen from the corresponding results in ASP \cite{Eiter2007Semantical}. 
For ASP programs, it has known that deciding uniform equivalence is $\Pi_2^p$-complete. 
Here, we show the upper bound of deciding semi-uniform equivalence. 
\begin{lem}
\label{lem:complexity-ue-checking}
For an \lpmln program $P$ and an SE-interpretation $(X,Y)$, deciding whether $(X, Y) \in \lue{P}$ is co-NP-complete. 
\end{lem}

\begin{proof}
As we know, checking whether $(X, Y) \in \lse{P}$ is in polynomial time.  
For the UE-model checking, if $X \subset Y$, we need to check there does not exist an interpretation $X'$ such that $X \subset X' \subset Y$ and $X' \models \lglred{P}{Y}$. 
	
\textbf{Hardness.}
Recall the proof of co-NP-hardness of SE-model checking of Theorem \ref{thm:complexity-se}, 
it is easy to check that a CNF $F$ is a tautology iff for any total interpretation $X$, $(X, X \cup \{b\})$ is an SE-model of $\psi_2(F)$. 
By the definition of UE-models, we have $(X, X \cup \{b\})$ is a UE-model of $\psi_2(F)$, 
therefore, the problem of checking tautology can be reduced to the problem of checking UE-model in polynomial time. 
Above results prove the co-NP-hardness of UE-model checking in \lpmlnend.

\textbf{Membership.}
As shown in \cite{Eiter2007Semantical}, 
the UE-model checking can be reduced to the problem of checking propositional entailment of the formula $\lglred{P}{Y} \cup X \cup Y_\subset \models X_=$, where 
\begin{equation}
	Y_\subset = \{\leftarrow y. ~|~ y \in U - Y\} \cup \{ \leftarrow y_1, \ldots, y_{n}. ~|~ y_i \in Y \text{ and } 1 \leq  i \leq |Y| \}
\end{equation}
\begin{equation}
	X_= = X \cup \{ \leftarrow x. ~|~  x \in U - X \}
\end{equation}
and $U$ is the universe of literals, i.e. $\olit{P} \subseteq U$.
Since the problem of checking propositional entailment is co-NP-complete, 
above results show the co-NP-membership of UE-model checking. 
\end{proof}
For the semi-uniform equivalence checking, we consider a complementary problem. 
To show that $P$ and $Q$ are not semi-uniformly equivalent, we can guess an SE-model $(X, Y)$ such that $(X, Y)$ is a UE-model of exactly one of the programs $P$ and $Q$. 
By Lemma \ref{lem:complexity-ue-checking}, the guess for $(X, Y)$ can be verified in polynomial time with the help of an NP oracle, 
which shows the $\Pi_2^p$-membership of deciding semi-uniform equivalence. 

\begin{thm}
\label{thm:complexity-ue}
For \lpmln programs $P$ and $Q$, deciding $P \equiv_{u,s} Q$ is  in $\Pi^p_2$.
\end{thm}

For the p-uniform equivalence checking, we have not found a sufficient and necessary condition for the characterization, 
therefore, it is not a good time to discuss its computational complexity.

\subsection{Discussion}
To prove the co-NP-membership of checking semi-strong equivalence, 
we present a translation from \lpmln programs to a propositional formula. 
In this subsection, we show a potential application of the translation. 

The notion of strong equivalences can be used to simplify logic programs, 
but the SE-model based strong equivalence checking is highly complex in computation. 
An available way to improve the checking is to find some 
syntactic conditions that decide the strong equivalence. 
In the field of ASP, there are some such kinds of syntactic conditions. 
For example, if the positive and negative bodies of an ASP rule $r$ have the same literals, i.e. $b^+(r) \cap b^-(r) \neq \emptyset$, the rule $r$ is strongly equivalent to the empty program $\emptyset$. 
Particularly, Lin and Chen \cite{Lin2005Discover} present a method to discover syntactic conditions that can be used to decide the strong equivalence of ASP. 
A main part of the method is a translation from ASP programs to a propositional formula, 
therefore, Lin and Chen's method can adapt to \lpmln by using the translation presented in this paper. 
Based on Lin and Chen's method and our translation, 
we believe there is a potential method to discover syntactic conditions that decide the semi-strong equivalence, 
which could be a next work of the paper.

\section{Relating to Other Logic Formalisms}
\label{sec:other-se}
Among the extensions of ASP, the strong equivalences for
ASP with weak constraints (\aspwcend) and ASP with ordered disjunction (LPOD) have been investigated \cite{Eiter2007ASPbpa,Faber2008SELPOD}. 
In this section, we investigate the relationships among the strong equivalences for \lpmlnend, \aspwcend, and LPOD.

\subsection{ASP with Weak Constraints}
A weak constraint is a kind of soft constraint of ASP that can be violated with a penalty, which is a part of the standard ASP language, 
and is introduced to represent the preferences among stable models \cite{Buccafurri2000Weakconstraint,Calimeri2012ASPcore}. 
An ASP program containing weak constraints is called an \aspwc program. 
Eiter et al.  have studied the strong equivalence between \aspwc programs \cite{Eiter2007ASPbpa}. 
In this section, we show how to characterize the strong equivalences for \aspwc by the sp-strong equivalence for \lpmlnend. 

Firstly, we review the semantics of \aspwcend. 
A weak constraint $r$ is a rule of the form 
\begin{equation}
\label{eq:weak-contraint}
:\sim l_{1}, ..., ~l_m, ~not~ l_{m+1}, ...,~not ~ l_n. ~ [penalty ~:~ level]
\end{equation}
where $l_i$s $(i \leq i \leq n)$ are literals,  $penalty$ is a real number denoting the cost of violating the constraint, 
and $level$ is a non-negative integer denoting the level of the penalization. 
In rest of the paper, we only consider the weak constraints of the form 
\begin{equation}
\label{eq:weak-constraint-simple}
:\sim l_{1}, ..., ~l_m, ~not~ l_{m+1}, ...,~not ~ l_n. ~ [penalty]
\end{equation} 
since the levels of penalization can be complied into penalties \cite{Eiter2007ASPbpa}.
For a weak constraint $r$ of the form \eqref{eq:weak-constraint-simple}, we use $pe(r)$ to denote the penalty associated with the rule. 
For an \aspwc program $P$, by $P^r$ and $P^c$, we denote the sets of plain ASP rules and weak constraints of $P$, respectively. 
For an interpretation $X$, 
by $WC(P, X)$, we denote the weak constraints of $P$ that are violated by $X$. 
An interpretation $X$ is a stable model of $P$ if $X$ is a stable model of the program $P^r$, 
and the penalty of $X$ w.r.t. $P$ is defined as 
\begin{equation}
Penalty(P, X) = \sum_{r \in WC(P, X)} pe(r)
\end{equation}
An optimal stable model of $P$ is a stable model of $P$ with the minimum penalty. 
Eiter et al. defined the strong equivalence for \aspwc programs as follows.

\begin{defi}
\label{def:se-aspwc}
Two \aspwc programs $P$ and $Q$ are strongly equivalent, if for any \aspwc program $R$, 
$P \cup R$ and $Q \cup R$ has the same stable models, 
and for any stable models $X$ and $Y$ of $P \cup R$,  $Penalty(P \cup R, X) - Penalty(P \cup R, Y) = Penalty(Q \cup R, X) - Penalty(Q \cup R, Y)$. 
\end{defi}

A characterization for the strong equivalence between \aspwc programs is shown as follows. 

\begin{lemC}[{\cite[Lemma 23]{Eiter2007ASPbpa}}]
\label{lem:se-aspwc-eiter}
Two \aspwc programs $P$ and $Q$ are strongly equivalent, iff $\ase{P^r} = \ase{Q^r}$, 
and for any interpretations $X$ and $Y$  satisfying $P^r \cup Q^r$,  $Penalty(P, X) - Penalty(P, Y) = Penalty(Q, X) - Penalty(Q, Y)$. 
\end{lemC}

Now, we show that the strong equivalence for \aspwc programs can be characterized by the sp-strong equivalence for \lpmln programs. 
For an \aspwc program $P$, its \lpmln translation is $\aspwctrans{P} = \aspwctranshard{P} \cup \aspwctranssoft{P}$, 
where $\aspwctranshard{P}$ and $\aspwctranssoft{P}$ are 
\begin{equation}
\aspwctranshard{P} = \{\alpha ~:~ r ~|~ r \in P - P^c\} 
\end{equation}
\begin{equation}
\aspwctranssoft{P} = \{w ~:~ \leftarrow body(r). ~|~  :\sim body(r). ~[w] ~ \in P^c\}
\end{equation}
For an interpretation $X$, it is easy to check that $X$ is a stable model of $P$ iff $X$ is a soft stable model of $\aspwctrans{P}$, and $X$ is an optimal stable model of $P$ iff $X$ is a most probable stable model of $\aspwctrans{P}$. 
More specifically, we have the following proposition.

\begin{prop}
\label{prop:aspwc-lpmln}
For an \aspwc program $P$ and its \lpmln translation $\aspwctrans{P}$, we have $\assm{P} = \ssm{\aspwctrans{P}}$, and for each stable model $X$, we have 
\begin{equation}
Penalty(P,X) = ln\left(W(\aspwctranssoft{P})\right) - ln\left( W_s(\aspwctrans{P}, X) \right)
\end{equation}
\end{prop}

Obviously, the strong equivalence between \aspwc programs can be characterized by the sp-strong equivalence for \lpmlnend, which is shown in Theorem \ref{thm:aspwc-se}.

\begin{thm}
\label{thm:aspwc-se}
Two \aspwc programs $P$ and $Q$ are strongly equivalent, iff $\aspwctrans{P}$ and $\aspwctrans{Q}$ are p-strongly equivalent under the SSM semantics
\end{thm}
Theorem \ref{thm:aspwc-se} can be proven simply by showing the conditions in Theorem \ref{thm:aspwc-se} and Lemma \ref{lem:se-aspwc-eiter} are equivalent, 
which is straightforward by corresponding definitions. 
Actually, the if part of the proof can be derived by properties of the \lpmln translation for \aspwc program directly. 
For two \aspwc programs $P$ and $Q$, since $\aspwctrans{P}$ and $\aspwctrans{Q}$ are sp-strongly equivalent, 
we have for any \lpmln program $N$, $\aspwctrans{P} \cup N$ and $\aspwctrans{Q} \cup N$ have the same soft stable models, 
and there exists a constant $c$ such that for any soft stable model $X \in \ssm{\aspwctrans{P} \cup N}$, $W_s(\aspwctrans{P} \cup N, X) = c * W_s(\aspwctrans{Q} \cup N, X)$. 
By Proposition \ref{prop:aspwc-lpmln}, we have for any \aspwc program $R$, 
$P \cup R$ and $Q \cup R$ have the same stable models, 
and for a stable model $X \in \assm{P \cup R}$, 
we have $Penalty(P \cup R, X) = ln \left(W(\aspwctranssoft{P \cup R})\right) - ln\left(W_s(\aspwctrans{P \cup R}, X)\right) = ln \left(W(\aspwctranssoft{Q \cup R})\right) - ln\left(W_s(\aspwctrans{Q \cup R}, X)\right) - ln(c) = Penalty(Q \cup R, X) - ln(c)$, 
which means for any stable models $X$ and $Y$ of $P \cup R$, we have $Penalty(P \cup R, X) - Penalty(Q \cup R, X) = Penalty(P \cup R, Y) - Penalty(Q \cup R, Y) $. 
Therefore, $P$ and $Q$ are strongly equivalent under the ASP semantics. 

\begin{exa}
\label{ex:aspwc-se}
Consider following \aspwc programs $P$
\begin{eqnarray}
& a \vee b. \\
& :\sim a. ~[1]
\end{eqnarray}
and $Q$
\begin{eqnarray}
& a \vee b. \\
& :\sim not ~a. ~[-1]
\end{eqnarray}
By the  definition of \lpmln translation $\tau^c$ for \aspwcend, $\aspwctrans{P}$ is 
\begin{eqnarray}
\alpha &:& a \vee b. \\
1 &:& \leftarrow a. 
\end{eqnarray}
and $\aspwctrans{Q}$ is 
\begin{eqnarray}
\alpha &:& a \vee b. \\
-1 &:& \leftarrow not ~a. 
\end{eqnarray}
It is easy to check that programs $P$, $Q$, $\aspwctrans{P}$, and $\aspwctrans{Q}$ have the same (soft) SE-models, which is shown in Table \ref{tab:aspwc-se}. 
From the table, it can be observed that $P$ and $Q$ are strongly equivalent, and $\aspwctrans{P}$ and $\aspwctrans{Q}$ are sp-strongly equivalent, 
which shows the relationship between the strong equivalences of \aspwc and \lpmlnend. 
\begin{table}
\centering
\caption{Computing Results in Example \ref{ex:aspwc-se}}
\label{tab:aspwc-se}
\def\arraystretch{1.3} 
\begin{tabular}{ccccc}
\hline
SE-model $(X, Y)$ & $(\{a\}, \{a\})$ & $(\{b\}, \{b\})$ & $(\{a\}, \{a,b\})$ & $(\{b\}, \{a,b\})$ \\ 
\hline
$Penalty(P, Y)$ & $1$ & $0$ & $1$ & $1$ \\ 
$Penalty(Q, Y)$ & $0$ & $-1$ & $0$ & $0$ \\ 
$W_s(\aspwctrans{P}, Y)$ & $e^0$ & $e^1$ & $e^0$ & $e^0$ \\ 
$W_s(\aspwctrans{Q}, Y)$ & $e^{-1}$ & $e^0$ & $e^{-1}$ & $e^{-1}$ \\ 
\hline
\end{tabular}
\end{table}
\end{exa}

\subsection{LPOD}
LPOD is another way to represent preferences over stable models by assigning priority to literals in the head of a rule \cite{Brewka2002LPOD,Brewka2006Preference}. 
An LPOD program $P$ consists of two parts: the regular part $P^r$ and the ordered disjunction part $P^o$, 
where $P^r$ is an ASP program consisting of rules of the form \eqref{eq:asp-rule-form}, and $P^o$ is a finite set of LPOD rules $r$ of the form \eqref{eq:ordered-disjunction-rule},
\begin{equation}
\label{eq:ordered-disjunction-rule}
h_1 \times ... \times h_n \leftarrow b^+(r), ~not~ b^-(r).
\end{equation}
where $h_i$s ($1 < i \leq n$) are literals that differ from each other.
By $o(r)$ we denote the number of literals occurred in the head of an LPOD rule $r$, i.e. $o(r) = |h(r)|$.
An LPOD rule $r$ of the form \eqref{eq:ordered-disjunction-rule} means if the body of $r$ is true, 
for any positive integers $i < j$, we prefer to believe $h_i$ rather than $h_j$, 
and if we believe $h_i$, it is not necessary to believe $h_j$. 

For an LPOD rule $r$, its \textit{i-th option} ($1 \leq i \leq o(r)$), denoted by $r^i$, is defined as 
	\begin{equation}
	\label{eq:i-th-option}
	h_i \leftarrow b^+(r), ~not~ b^-(r), ~not ~ h_1, ..., ~not ~h_{i-1}.
	\end{equation}
A \textit{split program} of an LPOD program $P$ is obtained from $P$ by replacing each rule in $P^o$  with one of its options. 
An interpretation $X$  is a \textit{candidate stable model} of $P$ if it is a stable model of a split program of $P$. 
By $\csm{P}$, we denote the set of all candidate stable models of $P$. 
The \textit{satisfaction degree} $deg(r,X)$ of an interpretation $X$ w.r.t an LPOD rule $r$ is defined as 
\begin{equation}
\label{eq:lpod-satisfaction-degree}
deg(r,X) = \begin{cases}
1, & \text{if } X \not\models b(r); \\ 
min\{k ~|~ h_k \in h(r) \cap X\}, & \text{if } X \models b(r) \text{ and } X \models h(r)\\
o(r), & otherwise.
\end{cases}
\end{equation}
And the \textit{satisfaction degree} $deg(P,X)$ of an interpretation $X$ w.r.t. an LPOD program $P$ is defined as the sum of satisfaction degrees of $X$ w.r.t. LPOD rules in $P^o$, i.e. $deg(P,X) = \sum_{r \in P^o} deg(r,X)$. 
For a candidate stable model $X$ of $P$, by $X^i(P)$ we denote the set of LPOD rules in $P^o$ that are satisfied by $X$ at degree $i$. 
Based on the notion of satisfaction degree, for two candidate stable models $X$ and $Y$ of $P$, Brewka \cite{Brewka2006Preference} introduces four preference criteria:
\begin{enumerate}
	\item \textbf{Cardinality-Preferred:} $X$ is cardinality-preferred to $Y$, denoted by $X >^c Y$, if there is a positive integer $i$ such that $|X^i(P)| > |Y^i(P)|$, and $|X^j(P)| = |Y^j(P)|$ for all $j < i$;
	
	\item \textbf{Inclusion-Preferred:} $X$ is inclusion-preferred to $Y$, denoted by $X >^i Y$, if there is a positive integer $i$ such that $Y^i(P) \subset X^i(P)$, and $X^j(P) = Y^j(P)$ for all $j < i$;
	
	\item \textbf{Pareto-Preferred:} $X$ is pareto-preferred to $Y$, denoted by $X >^p Y$, if there is a rule $r \in P_o$ such that $deg(r,X) < deg(r,Y)$, and there does not exist a rule $r \in P_o$ such that $deg(r,Y) < deg(r,X)$.
	
	\item \textbf{Penalty-Sum-Preferred:} $X$ is penalty-sum-preferred to $Y$, denoted by $X >^{ps} Y$, if $deg(P,X) < deg(P,Y)$.
\end{enumerate}
For each $pr \in \{c,i,p,ps\}$, a candidate stable model $X$ of $P$ is a $pr$-preferred stable model if there is no candidate stable model $Y$ of $P$ such that $Y >^{pr} X$. 

To investigate the strong equivalence for LPOD, we use the definition of candidate stable models for LPOD presented by Faber et al. \cite{Faber2008SELPOD}.
For an LPOD program $P$ and an interpretation $X$, the minimum split program $SP(P, X)$ of $P$ w.r.t. $X$ is defined as 
\begin{equation}
SP(P, X) = P^r \cup \{r^{deg(r, X)} ~|~ r \in P^o\}
\end{equation}
Faber et al. show that an interpretation $X$ is a candidate stable model of $P$ iff $X$ is a stable model of $SP(P, X)$. 
Based on the result, Faber et al. investigate the non-preferential strong equivalence for LPOD programs by generalizing the notion of SE-models in ASP.

\begin{defiC}[{\cite[Definition 2]{Faber2008SELPOD}}]
Two LPOD programs $P$ and $Q$ are strongly equivalent, denoted by $P \equiv^\times_{s} Q$, if for any LPOD program $R$, $\csm{P \cup R} = \csm{Q \cup R}$.
\end{defiC}

\begin{defiC}[{\cite[Definition 3]{Faber2008SELPOD}}]
For an LPOD program $P$, an SE-interpretation $(X, Y)$ is an SE-model of $P$, if there exists a split program $P'$ such that $Y \models SP(P, Y)$ and $X \models SP(P, Y)^Y$.  
By $SE^\times(P)$, we denote the set of all SE-models of $P$.
\end{defiC}

\begin{lemC}[{\cite[Theorem 6]{Faber2008SELPOD}}]
\label{lem:lpod-se-faber}
Two LPOD programs $P$ and $Q$ are strongly equivalent iff $\cse{P} = \cse{Q}$.
\end{lemC}

Now, we show that the strong equivalence for LPOD can be characterized by the semi-strong equivalence under the SSM semantics. 
Following the method in the previous subsection, our approach is outlined as follows: 
(1) find a translation from LPOD to \lpmlnend; (2) show two LPOD programs are strongly equivalent iff their \lpmln translations are semi-strongly equivalent under the SSM semantics. 

For an LPOD program $P$, an \lpmln translation $\lpodtrans{P}$ of $P$ consists of three parts, i.e. $\lpodtrans{P} = \tau^\times_1(P^r) \cup \tau^\times_2(P^o) \cup \tau^\times_3(P^o)$, 
where 
\begin{itemize}
	\item $\tau^\times_1(P) = \{\alpha ~:~ r ~|~ r \in P^r\}$; 
	\item $\tau^\times_2(P) = \{1 ~:~ r^k ~|~ r \in P^o, ~1 \leq k \leq o(r) \}$, $r^k$ is the k-th option of an LPOD rule $r$; and 
	\item $\tau^\times_3(P) = \{ \alpha ~:~ \leftarrow body(r), ~not~ h_1, ..., ~not ~h_{o(r)}  ~|~  r \in P^o\}$.
\end{itemize}
Wang et al. show that there is a one-to-one map between the candidate stable models of an LPOD program $P$ and the stable models of $\lpodtrans{P}$ \cite{Wang2018Preference}, 
more specifically, we have the following proposition. 

\begin{prop}
\label{prop:lpmln-lpod}
For an LPOD program $P$ and its \lpmln translation $\lpodtrans{P}$, we have $\csm{P} = \ssm{\lpodtrans{P}}$, and for each candidate stable model $X \in \csm{P}$, we have 
\begin{equation}
deg(r, X) = ln(W(\lpodtranssoft{\{r\}})) - ln(W_s(\lpodtrans{\{r\}}, X)) + 1, \text{ for each } r \in P^o, 
\end{equation}
\begin{equation}
deg(P, X) = ln(W(\lpodtranssoft{P})) - ln(W_s(\lpodtrans{P}, X)) + |P^o|
\end{equation}
\end{prop}

By the translation $\tau^\times$, the strong equivalence for LPOD programs can be investigated in \lpmlnend, which is shown as follow.

\begin{thm}
\label{thm:lpod-se-lpmln}
Two LPOD programs $P$ and $Q$ are strongly equivalent, iff their \lpmln translations $\lpodtrans{P}$ and $\lpodtrans{P}$ are semi-strongly equivalent under the SSM semantics. 
\end{thm}

\begin{proof}
We prove Theorem \ref{thm:lpod-se-lpmln} by showing that for an LPOD program $P$ and its \lpmln translation $\lpodtrans{P}$, $\cse{P} = \sse{\lpodtrans{P}}$. 
Suppose $(X, Y)$ is an SE-interpretation such that $Y \models P$, 
by the definition of $\tau^\times$, it is easy to check that $Y \models \tau_1^\times(P) \cup \tau_3^\times(P)$ i.e. $Y \models \lpodtrans{P}^h$. 
For an LPOD rule $r$, there are two cases: (1) $Y \not\models b(r)$; (2) $Y \models b(r)$ and $Y \models h(r)$. 

\textbf{Case 1.}
If $Y \not\models b(r)$, the satisfaction degree is $deg(r, I) = 1$, 
therefore, the minimum split program of $\{r\}$ w.r.t. $Y$ is $SP(\{r\}, Y) = \{r^1\}$ 
and the \lpmln reduct of $\lpodtrans{\{r\}}$ w.r.t. $Y$  is  $\left(\lpodtrans{\{r\}}\right)_Y = \{ 1 ~:~ r^i ~|~ 1 \leq 
i \leq o(r)\} $. 
If $Y \cap b^-(r) \neq \emptyset$, then we have $SP(\{r\}, Y)^Y = \emptyset$, and $\lglred{\left(\lpodtrans{\{r\}}\right)}{Y} = \emptyset$, 
which means for any proper subset $X$ of $Y$, $X \models SP(\{r\}, Y)^Y$ and $X \models \lglred{\left(\lpodtrans{\{r\}}\right)}{Y}$. 
If $b^-(r) \cap Y = \emptyset$, then we have $b^+(r) \not\subseteq Y$, $SP(\{r\}, Y)^Y = \{ h(r^1) \leftarrow b^+(r).  \}$.
For the \lpmln reduct $\left(\lpodtrans{\{r\}}\right)_Y$, there exists an integer $k$ such that $\lglred{\left(\lpodtrans{\{r\}}\right)}{Y} = \{ h(r^i) \leftarrow b^+(r). ~|~ 1 \leq i \leq k \}$. 
Note that for an LPOD rule of the form \eqref{eq:ordered-disjunction-rule}, if $h(r) \cap Y \neq \emptyset$, then $ k = min\{ j ~|~ h_j \in h(r) \cap Y\} $, otherwise, $k = o(r)$. 
It is easy to check that for any proper subset $X$ of $Y$, $X \models SP(\{r\}, Y)^Y$ and $X \models \lglred{\left(\lpodtrans{\{r\}}\right)}{Y}$. 

\textbf{Case 2.}
If $Y \models b(r)$ and $Y \models h(r)$, suppose satisfaction degree is $deg(r, I) = k$, 
therefore, the minimum split program of $\{r\}$ w.r.t. $Y$ is $SP(\{r\}, Y) = \{r^k\}$ 
and the \lpmln reduct of $\lpodtrans{\{r\}}$ w.r.t. $Y$  is  $\left(\lpodtrans{\{r\}}\right)_Y = \{ 1 ~:~ r^i ~|~ k \leq 
i \leq o(r)\} $. 
Since $Y \models b(r)$, we have $SP(\{r\}, Y)^Y = \{ h(r^k) \leftarrow b^+(r).  \}$ and $\lglred{\left(\lpodtrans{\{r\}}\right)}{Y} = \{ h(r^k) \leftarrow b^+(r). \}$. 
Therefore, for any proper subset $X$ of $Y$, $X \models SP(\{r\}, Y)^Y$ iff $X \models \lglred{\left(\lpodtrans{\{r\}}\right)}{Y}$. 

Combining above results, we have shown that an SE-interpretation $(X, Y)$ is an SE-model of LPOD program $P$ iff $(X, Y)$ is a soft SE-model of the translation $\lpodtrans{P}$, 
therefore, Theorem \ref{thm:lpod-se-lpmln} is proven. 
\end{proof}
\begin{exa}
\label{ex:lpod-se}
Consider LPOD programs $P$ 
\begin{eqnarray}
& a \times b. \\
& b \times a. \\
& c. 
\end{eqnarray}
and $Q$ 
\begin{eqnarray}
& a \times b. \\
& b \times c. \\
& c. 
\end{eqnarray}
By the definition of \lpmln translation $\tau^\times$ for LPOD, 
$\lpodtrans{P}$ is 
\begin{eqnarray}
 \alpha &:& \leftarrow not ~ a, ~not~ b. \\
 1 &:& a.  \\
 1 &:& b \leftarrow not~ a. \\
 1 &:& b.  \\
 1 &:& a \leftarrow not~ b. \\
 \alpha &:& c. 
\end{eqnarray}
and $\lpodtrans{Q}$ is 
\begin{eqnarray}
\alpha &:& \leftarrow not ~ a, ~not~ b. \\
1 &:& a.  \\
1 &:& b \leftarrow not~ a. \\
\alpha &:& \leftarrow not ~ b, ~not~ c. \\
1 &:& b.  \\
1 &:& c \leftarrow not~ b. \\
\alpha &:& c. 
\end{eqnarray}
It is easy to check that the programs $P$, $Q$, $\lpodtrans{P}$, and $\lpodtrans{Q}$ have the same (soft) SE-models, and the satisfaction / weight degrees of each SE-models are shown in Table \ref{tab:lpod-se}. 
Therefore, $P$ and $Q$ are strongly equivalent, and $\lpodtrans{P}$ and $\lpodtrans{Q}$ are semi-strongly equivalent. 
\begin{table}
	\centering
	\caption{Computing Results in Example \ref{ex:lpod-se}}
	\label{tab:lpod-se}
  \def\arraystretch{1.3} 
	\begin{tabular}{ccccc}
		\hline
		SE-model (X, Y) & $(\{a, c\}, \{a, c\})$ & $(\{b, c\}, \{b, c\})$ & $(\{a, b, c\}, \{a, b, c\})$ \\ 
		\hline
		$deg(P, Y)$ & $3$ & $3$ & $2$ \\ 
		$deg(P, Y)$ & $3$ & $3$ & $2$ \\ 
		$W_s(\lpodtrans{P}, Y)$ & $e^{3}$ & $e^{3}$ & $e^{4}$ \\ 
		$W_s(\lpodtrans{Q}, Y)$ & $e^{3}$ & $e^{3}$ & $e^{4}$ \\ 
		\hline
	\end{tabular}
\end{table}
\end{exa}

Besides non-preferential strong equivalence, Faber et al. also investigate the strong equivalence under the cardinality, inclusion, and pareto preference criteria, 
which cannot be investigated in \lpmln without introducing new notions. 
But, by Proposition \ref{prop:lpmln-lpod}, the strong equivalence under the penalty-sum criterion can be characterized in \lpmlnend, 
which has not been discussed by Faber et al. 

\begin{defi}
Two LPOD programs $P$ and $Q$ are ps-strongly equivalent, denoted by $P \equiv_{s, ps}^\times Q$, 
if for any LPOD program $R$, $\csm{P \cup R} = \csm{Q \cup R}$, and for any candidate stable model $X$ and $Y$, 
$deg(P \cup R, X) - deg(P \cup R, Y) = deg(Q \cup R, X) - deg(Q \cup R, Y)$. 
\end{defi}

\begin{thm}
\label{thm:lpod-psse-lpmln}
Two LPOD programs $P$ and $Q$ are strongly equivalent, iff their \lpmln translations $\lpodtrans{P}$ and $\lpodtrans{P}$ are sp-strongly equivalent. 
\end{thm}
The proof of Theorem \ref{thm:lpod-psse-lpmln} is straightforward based on Theorem \ref{thm:lpod-se-lpmln} and Proposition \ref{prop:lpmln-lpod}, 
therefore, we omit the detail for brevity.
Theorem \ref{thm:lpod-psse-lpmln} shows that the ps-strong equivalence for LPOD can be characterized by the sp-strong equivalence in \lpmlnend. 
Recall LPOD programs $P$ and $Q$ in Example \ref{ex:lpod-se}, 
it is easy to check that $P$ and $Q$ are ps-strongly equivalent by Theorem \ref{thm:lpod-psse-lpmln}.

Now, we have shown that the strong equivalences for LPOD programs can be characterized by translating them into \lpmln programs. 
Actually, we can prove the only-if part of Theorem \ref{thm:lpod-se-lpmln} directly, just as the method used in proving Lemma \ref{lem:lpmln-sm-strong-equiv} and Lemma \ref{lem:semi-pse-ssm}. 
By the definition of $\tau^\times$, regular ASP rules in an LPOD program $P$ are turned into hard rules by directly assigning a weight ``$\alpha$'', 
therefore, the direct proof of only-if part of Theorem \ref{thm:lpod-se-lpmln} is exactly the same as the corresponding proof of Lemma \ref{lem:semi-pse-ssm}.

\subsection{Discussion}

In the last few years, researchers have investigated the relationships among \lpmln and several logic formalisms such as ASP, MLN, LPOD, P-log \cite{Baral2009Probabilistic}, ProbLog \cite{DeRaedt2007problog} etc. 
On the one hand, these results provide a translation based approach to implementing \lpmln solvers, 
which is used in several implementations such as LPMLN2ASP, LPMLN2MLN \cite{Lee2017ComputingLpmln}, and LPMLN-Models \cite{Wu2018LPMLNModels}. 
On the other hand, these relationships help us to investigate new properties of a logic formalism by using existing results in other logic formalisms, 
such as the investigations of splitting set theorems \cite{Wang2018Splitting} and weight learning of \lpmln \cite{Lee2018WeightLearning}. 

Under the circumstance, 
one may ask whether the strong equivalence of other formalisms can be investigated in \lpmlnend. 
From the discussion of \aspwc and LPOD, we can observe that the key of the question is the translation from a logic program to an \lpmln program. 
There are some common properties of the translations $\tau^c$ and $\tau^\times$ used for \aspwc and LPOD.  
Firstly, a translation should be \textit{semantics-preserving}, i.e. original program and the translated program have the same inference results. 
Obviously, this is a necessary property for a translation, 
and both $\tau^c$ and $\tau^\times$ are semantics-preserving. 
Secondly, a translation should be \textit{modular}. 
A translation $\tau^*$ is modular, if for logic programs $P$ and $Q$, we have $\tau^*(P \cup Q) = \tau^*(P) \cup \tau^*(Q)$. 
By investigating the strong equivalences for \aspwc and LPOD, we have shown that for a logic program $P$, if there is a modular \lpmln translation $\tau'$ such that 
the semantics of $P$ can be characterized by $\tau'(P)$, 
then we can at least find a sufficient condition for characterizing the strong equivalence of the program $P$. 
In other words, if we cannot find a modular translation for the program, 
then its strong equivalence may not be investigated via using the results obtained in this paper. 

In addition, the \lpmln translations $\tau^c$ and  $\tau^\times$ for \aspwc and LPOD have other good properties.
For example, for a logic program $P$ and a translation $\tau^*(P)$, $\tau^*$ is called \textit{fixed}, if $\tau^*(P)$ does not introduce new atoms, i.e. $lit(P) = lit(\tau^*(Q))$. 
It is easy to check that both the translations $\tau^c$ and  $\tau^\times$ are fixed. 
But it is still unknown whether the property affects the translation based investigation of strong equivalences.  
For P-log programs, there exists a translation $\tau^p$ from P-log to \lpmln \cite{Lee2017lpmln}, 
and there also exists a translation $\tau^m$ from \lpmln to P-log \cite{Balai2016realtionship}. 
Both of $\tau^p$ and $\tau^m$ are semantics-preserving and modular, but they are not fixed. 
Intuitively, since there is a one-to-one map between P-log programs and \lpmln programs, 
it is very likely the strong equivalence for P-log can be characterized by the strong equivalence for \lpmlnend. 

In a word, our results provide a method to investigate the strong equivalence for some logic formalisms by translating them into \lpmln programs. 
And for the method, 
some properties of translations are important, such as the semantics-preserving and modular properties. 
But it is unclear whether newly introduced literals in the translation affect the translation based investigation of strong equivalences.

\section{Simplifying \texorpdfstring{\lpmln }{LPMLN }Programs}
\label{sec:simplify}
Program simplification is an important technology for improving the implementations of logic programs. 
In this section, we investigate the simplification of \lpmln programs via introducing the notions of semi-valid and valid rules, 
which are two kinds of redundant \lpmln rules based on the semi-strong and p-strong equivalences. 
Firstly, we present an algorithm to simplify and solve \lpmln programs by eliminating these redundant \lpmln rules.  
Then, to decide the redundant rules efficiently, we present some syntactic conditions that characterize the semi-valid and valid \lpmln rules. 

According to whether an \lpmln rule is semi-strongly or p-strongly equivalent to the empty program $\emptyset$, the notions of semi-valid and valid \lpmln rules are defined as follows. 
\begin{defi}
\label{def:valid-rule}
An \lpmln rule $w:r$ is called semi-valid, if $w:r$ is semi-strongly equivalent to $\emptyset$; 
the rule is called valid, if $w:r$ is p-strongly equivalent to $\emptyset$.
\end{defi}

Obviously, a valid \lpmln rule can be eliminated from any \lpmln programs, 
while a semi-valid \lpmln rule cannot. 
By the definition, eliminating a semi-valid \lpmln rule does not change the stable models of original programs, 
but changes the probability distributions of the stable models. 
Furthermore, it may change the probabilistic stable models of original programs, 
which can be observed from Example \ref{ex:sm-valid-rules}.

\begin{exa}
\label{ex:sm-valid-rules}
Consider three \lpmln programs $P = \{ \alpha : a \leftarrow a.  \}$, $Q = \{ \alpha : \leftarrow a.  \}$, and $R = \{1 : a. \}$. 
It is easy to check that rules in $P$ and $Q$ are valid and semi-valid, respectively.  
Table \ref{tab:sm-valid-rules} shows the stable models and their probability degrees of \lpmln programs $P \cup R$, $Q \cup R$, and $R$. 
It is easy to check that eliminating the rule of $P$ from program $P \cup R$ does not affect the inference results of $P \cup R$. 
By contrast, eliminating the rule of $Q$ from $Q \cup R$ changes both of the MAP and MPD inference results of $Q \cup R$.

\begin{table}
\centering
\caption{Computing Results in Example \ref{ex:sm-valid-rules}}
\label{tab:sm-valid-rules}
\def\arraystretch{1.3} 
\begin{tabular}{cccc}
\hline
Stable Model $S$ & $Pr(P \cup R, S)$ & $Pr(Q \cup R, S)$ & $Pr(R, S)$  \\ 
\hline
$\emptyset$ & $0.27$ & $1$ & $0.27$ \\ 
$\{a\}$ & $0.73$ & $0$ & $0.73$ \\ 
\hline 
\end{tabular}
\end{table}
\end{exa}

Algorithm \ref{alg:solver} provides a framework to simplify and solve \lpmln programs based on the notions of semi-valid and valid \lpmln rules. 
Firstly, simplify an \lpmln program $P$ by removing all semi-valid and valid rules (line 2 - 8). 
Then, compute the stable models of the simplified \lpmln program via using some existing \lpmln solvers, such as LPMLN2ASP, LPMLN2MLN \cite{Lee2017ComputingLpmln}, and LPMLN-Models \cite{Wu2018LPMLNModels} ect. 
Finally, compute the probability degrees of the stable models w.r.t. the simplified program and all semi-valid rules (line 9 - 12). 
The correctness of the algorithm can be proven by corresponding definitions.

\begin{algorithm}
\caption{Simplify and Solve \lpmln Programs}
\KwIn{an \lpmln program $P$}
\KwOut{stable models of $P$ and their probability degrees}
\label{alg:solver}
$Q = \emptyset$, $P' = P$ \;
\ForEach{$w:r \in P$}{ 
	\uIf{$w:r$ is valid}{
		$P' = P' -\{w:r\}$ \;
	}
	\uElseIf{$w:r$ is semi-valid}{
		$Q = Q \cup \{w:r\}$ \;
		$P' = P' -\{w:r\}$ \;
	}
}

$\lsm{P} = call\text{-}lpmln\text{-}solver(P')$\;

\ForEach{$X \in \lsm{P}$}{
	$W'(P, X) = exp\left(\sum_{w:r \in P' \cup Q \text{ and } X \models w:r} ~w \right)$ \;
}

Compute probability degrees for each stable model $X$ by Equation \eqref{eq:probability-sm} and $W'(P, X)$\;
\Return{$\lsm{P}$ and corresponding probability degrees}
\end{algorithm}

\begin{table}
\centering
\caption{Syntactic Conditions for Valid and Semi-valid \lpmln Rules}
\label{tab:syntactic-conditions}
\def\arraystretch{1.3} 
\begin{tabular}{ccc}
\hline
Name & Definition & Strong Equivalence \\
\hline
TAUT & $h(r) \cap b^+(r) \neq \emptyset$ & p, semi\\
CONTRA & $b^+(r) \cap b^-(r) \neq \emptyset$ & p, semi\\
CONSTR1 & $h(r) = \emptyset$ & semi \\
CONSTR2 & $h(r) \subseteq b^-(r)$ &  semi \\
CONSTR3 & $~~~~h(r) = \emptyset$, $b^+(r) = \emptyset$, and $b^-(r) = \emptyset~~~~$ & p, semi \\
\hline
\end{tabular}
\end{table}

In Algorithm \ref{alg:solver}, a crucial problem is to decide whether an \lpmln rule is valid or semi-valid. 
Theoretically, it can be done by checking the SE-models of each rule. 
However, the model-theoretical approach is highly complex in computation. 
Therefore, we investigate the syntactic conditions for the problem. 
Table \ref{tab:syntactic-conditions} shows five syntactic conditions for a rule $r$, 
where TAUT and CONTRA have been introduced to investigate the program simplification of ASP \cite{Osorio2001Equivalence,Eiter2004SimplifyingLP}, 
CONSTR1 means the rule $r$ is a constraint, and CONSTR3 is a special case of CONSTR1.
Rules satisfying CONSTR2 is usually used to eliminate constraints in ASP. 
For example, rule ``$\leftarrow a. $'' is equivalent to rule ``$p \leftarrow a, ~not ~p.$'', if the atom $p$ does not occur in other rules. 
Based on these conditions, the following theorems provide the characterizations for semi-valid and valid \lpmln rules.

\begin{thm}
\label{thm:semi-valid}
An \lpmln rule $w:r$ is semi-valid, iff the rule satisfies one of TAUT, CONTRA, CONSTR1, CONSTR2, and CONSTR3. 
\end{thm}

\begin{proof}
The if direction of Theorem \ref{thm:semi-valid} is straightforward. 
For each of TAUT, CONTRA, CONSTR1, CONSTR2, and CONSTR3, it can be verified by Lemma \ref{lem:lpmln-sm-strong-equiv} and the definition of SE-models. 
Here, we only present the proof of TAUT, the proofs of other conditions are similar. 

For the condition TAUT, let $P = \{w:r\}$, to prove the condition, we need to show that $\lse{P} = \lse{\emptyset}$.
It is easy to observe that for any SE-interpretation $(X,Y)$, $(X,Y) \in \lse{\emptyset}$, therefore, we only need to show that for any SE-interpretation $(X,Y)$, $(X,Y) \in \lse{P}$. We use proof by contradiction. 
Assume that there is an SE-interpretation $(X,Y)$ such that $(X,Y) \not\in \lse{P}$, 
by the definition, we have $X \not\models \lglred{P}{Y}$. 
For the program $P = \{w:r\}$, since $h(r) \cap b^+(r) \neq \emptyset$, we have $w:r$ can be satisfied by $Y$, which means $\overline{P_Y} = \overline{P}$. 
For the GL-reduct $\lglred{P}{Y}$, there are two cases: (1) $\lglred{P}{Y} = \emptyset$ and (2) $\lglred{P}{Y} \neq \emptyset$. 

\textbf{Case 1.} If $\lglred{P}{Y} = \emptyset$, we have $Y \cap b^-(r) \neq \emptyset$. 
It is obvious that $X \models \lglred{P}{Y}$, which contradicts with the assumption.

\textbf{Case 2.} If $\lglred{P}{Y} \neq \emptyset$, we have $Y \cap b^-(r) = \emptyset$. 
Since $X \not\models \lglred{P}{Y}$, we have $b^+(r) \subseteq X$ and $h(r) \cap X = \emptyset$, which means $h(r) \cap b^+(r) = \emptyset$. 
It contradicts with the condition TAUT.  

Combining above results, the condition TAUT is proven. 

To prove the only-if direction, we need to show that if a rule $w:r$  satisfies none of TAUT, CONTRA, CONSTR1, CONSTR2, and CONSTR3, 
then an \lpmln program $P = \{w:r\}$ is not semi-strongly equivalent to $\emptyset$, which means 
$h(r) \cap b^+(r) = \emptyset$, $b^+(r) \cap b^-(r) = \emptyset$, $h(r) \neq \emptyset$, and $h(r) \not\subseteq b^-(r)$.
There are four cases according to whether the body of $r$ is empty, 
we need to show that for each case, there exists an SE-interpretation that is not an SE-model of $P$.

\textbf{Case 1.}
If $b^+(r) = \emptyset$ and $b^-(r) = \emptyset$, $r$ is a fact. 
Obviously, $(\emptyset, h(r))$ is not an SE-model of $P$, 
therefore, $P$ is not semi-strongly equivalent to $\emptyset$.

\textbf{Case 2.}
If $b^+(r) = \emptyset$ and $b^-(r) \neq \emptyset$, 
it is easy to check that $(\emptyset, h(r) - b^-(r))$ is not an SE-model of $P$, 
therefore, $P$ is not semi-strongly equivalent to $\emptyset$.

\textbf{Case 3.}
If $b^+(r) \neq \emptyset$ and $b^-(r) = \emptyset$, 
it is easy to check that $(b^+(r), h(r) \cup b^+(r))$  is not an SE-model of $P$, 
therefore, $P$ is not semi-strongly equivalent to $\emptyset$.

\textbf{Case 4.}
If $b^+(r) \neq \emptyset$ and $b^-(r) \neq \emptyset$, 
it is easy to check that $(b^+(r), (h(r) - b^-(r)) \cup b^+(r))$  is not an SE-model of $P$, 
therefore, $P$ is not semi-strongly equivalent to $\emptyset$.

Combining above results, the only-if direction is proven. 
\end{proof}
\begin{thm}
\label{thm:valid}
An \lpmln rule $w:r$ is valid, iff one of the following conditions is satisfied 
\begin{itemize}
	\item[-] rule $w:r$ satisfies one of CONSTR1 and CONSTR2, and $w=0$; or
	\item[-] rule $w:r$ satisfies one of TAUT, CONTRA, and CONSTR3. 
\end{itemize}
\end{thm}

\begin{proof}
For the if direction, the proof is divided into three cases, in each case, we prove some conditions in Theorem \ref{thm:valid}. 

\textbf{Case 1.}
If $w:r$ is a rule of the form $CONSTR1$ or $CONSTR2$ and $w=0$. 
By Theorem \ref{thm:semi-valid}, $w:r$ is semi-valid. 
Since $w=0$, we have for any SE-models $(X, Y)$ of $\{w:r\}$, $W(\{w:r\}, (X, Y)) = e^0 = 1$, therefore, $w:r$ is valid. 

\textbf{Case 2.}
If $w:r$ is a rule of the form $TAUT$ or $CONSTRA$, it is easy to check that $w:r$ is semi-valid, and for any SE-model $(X, Y)$, $Y \models w:r$. 
Therefore, we have for any SE-model $(X, Y)$, $W(\{w:r\}, (X,Y)) = e^w$, which means $w:r$ is valid. 

\textbf{Case 3.}
If $w:r$ is a rule of the form  $CONSTR3$, it is easy to check that $w:r$ is semi-valid, and for any SE-model $(X, Y)$, $Y \not\models w:r$. 
Therefore, we have for any SE-model $(X, Y)$, $W(\{w:r\}, (X,Y)) = e^0 = 1$, which means $w:r$ is valid. 

For the only-if direction, we use proof by contradiction. 
Assume $w:r$ is a valid \lpmln rule satisfying none of conditions in Theorem \ref{thm:valid}, i.e. (1) $w:r$ does not satisfy all of $TAUT$, $CONSTRA$, $CONSTR1$, $CONSTR2$, and $CONSTR3$; (2) $w:r$ satisfies $CONSTR1$ or $CONSTR2$, and $w\neq 0$. 

\textbf{Case 1.}
If $w:r$ does not satisfy all of $TAUT$, $CONSTRA$, $CONSTR1$, $CONSTR2$, and $CONSTR3$, by Theorem \ref{thm:semi-valid}, $w:r$ is not semi-valid, therefore, $w:r$ is not valid, which contradicts with the assumption. 

\textbf{Case 2.}
If $w:r$ satisfies $CONSTR1$ or $CONSTR2$, and $w \neq 0$, 
by the definition of SE-models, there exist SE-models $(X_1, Y_1)$ and $(X_2, Y_2)$ such that $Y_1 \models w:r$ and $Y_2 \not\models w:r$, such as $Y_1 = h(r)$ and $Y_2 = b^+(r) - b^-(r)$. 
Therefore, we have $W(\{w:r\}, (X_1, Y_1)) = e^w$ and $W(\{w:r\}, (X_2, Y_2)) = e^0$, but $W(\{w:r\}, (X_1, Y_1)) = W(\{w:r\}, (X_2, Y_2)) = e^0$, 
which means $w:r$ is not valid and contradicts with the assumption. 

Combining above results, Theorem \ref{thm:valid} is proven. 
\end{proof}

Theorem \ref{thm:semi-valid} and Theorem \ref{thm:valid} can be used to check the validity of an \lpmln rule efficiently, 
which makes Algorithm \ref{alg:solver} an alternative approach to enhance \lpmln solvers.  
In addition, Theorem \ref{thm:semi-valid} and Theorem \ref{thm:valid} also contribute to the field of knowledge acquiring. 
On the one hand, although it is impossible that rules of the form TAUT, CONTRA, and CONSTR3 are constructed by a skillful knowledge engineer, 
these rules may be obtained by rule learning. 
Therefore, we can use TAUT, CONTRA, and CONSTR3 as heuristic information to improve the results of rule learning. 
On the other hand, it is worth noting that conditions CONSTR1, CONSTR2, and CONSTR3 mean the only effect of constraints in \lpmln is to change the probability distribution of inference results, 
which can be observed in Example \ref{ex:flattening-extension}. 
Therefore, for the problem modeling in \lpmlnend, 
we can encode objects and relations by \lpmln facts and rules, 
and adjust the certainty degrees of inference results by \lpmln constraints.

\section{Conclusion and Future Work}
\label{sec:conclusion}
In this paper, we investigate the notions of strong equivalences for \lpmln programs and study several properties of the notions from four aspects. 
First of all, we present the notion of p-ordinary equivalence for \lpmlnend, 
that is, two p-ordinarily equivalent \lpmln programs have the same stable models and the same probability distribution of their stable models, 
which means the programs have the same MAP and MPD inference results. 
Based on the p-ordinary equivalence, we present the notion of p-strong equivalence, 
that is, two p-strongly equivalent \lpmln programs are p-ordinary equivalent under any extensions. 
Then, we present a sufficient and necessary condition for characterizing the p-strong equivalence, i.e. the PSE-condition, 
which can be regarded as a generalization of the SE-model approach in ASP. 
Due to hard rules can be violated in the original \lpmln semantics, 
the necessity of the PSE-condition is quite difficult to prove. 
To this end, we introduce the notions of necessary extensions and flattening extensions, 
and show that the PSE-condition is necessary, since we can construct a necessary extension by selecting proper flattening extensions. 
Since the conditions of p-strong equivalence are somewhat strict, 
we study the notion of q-strong equivalence. 
Unfortunately, our results show that the q-strong equivalence is identical to the p-strong equivalence. 
Besides, we present a formal comparison between the notions and characterizations of strong equivalence present in this paper and Lee and Luo's work. 
It shows that the semi-strong equivalence and the structural equivalence are equivalent to each other, 
and the p-strong equivalence and the LL-strong equivalence are also equivalent to each other.

After the characterization, we further study the properties of the p-strong equivalence from four aspects. 
Firstly, we present two relaxed notions of the p-strong equivalence, i.e. the sp-strong and p-uniform equivalence, 
and discuss their characterizations, 
which are useful in many real-world scenarios such as decision making and knowledge graph based applications etc. 
Secondly, we analyze the computational complexities of deciding strong equivalences. 
Our results show that deciding all of the semi-strong, the p-strong, and the sp-strong equivalence is co-NP-complete, 
and deciding the semi-uniform equivalence is in $\Pi_2^p$. 
As a by-product, Lin and Chen's work \cite{Lin2005Discover} implies that our method of proving the complexity results could be used to discover syntactic conditions deciding semi-strong equivalence. 
Thirdly, we investigate the relationships among the strong equivalences for \lpmln and two important extensions of ASP:  \aspwc and LPOD.
Our results show that the strong equivalences for \aspwc and LPOD can be studied by translating them into \lpmlnend, 
which provides a viable way to study the strong equivalences for other logic formalisms such as ProbLog and P-log etc. 
Finally, we investigate the program simplification based on the p-strong equivalence. 
Specifically, we present a characterization for two kinds of redundant \lpmln rules, which can be used to improve \lpmln solvers.

For the future, we plan to continue the unsolved problems in this paper, 
i.e. the characterization of p-uniform equivalence and its computational complexity. 
And we will further investigate the strong equivalences for other logic formalisms by translating them into \lpmln programs. 
In addition, the investigations of approximate strong equivalences and the syntactic conditions deciding strong equivalences are also valuable topics of the field.

\section*{Acknowledgments}
We are grateful to the anonymous referees for their useful comments on the earlier version of this paper. 
The work was supported by the Pre-research Key Laboratory Fund for Equipment (Grant No. 6142101190304)
and the National Key Research and Development Plan of China (Grant No. 2017YFB1002801).

\bibliographystyle{alpha}
\bibliography{lpmln_equivalence_lmcs}

\end{document}